\documentclass[a4paper,noarxiv,twocolumn]{quantumarticle}
\usepackage[utf8]{inputenc}
\usepackage[english]{babel}
\usepackage[T1]{fontenc}

\usepackage{lipsum}
\usepackage{mdframed}
\usepackage{booktabs}
\usepackage{longtable}

\usepackage[numbers]{natbib}
\usepackage{bm}
\usepackage{graphicx}
\usepackage{amsmath}
\usepackage{amsfonts}
\usepackage{amssymb}
\usepackage{amsthm}
\usepackage{amstext}
\usepackage{amsbsy}
\usepackage{amsopn}
\usepackage{amscd}
\usepackage{amsxtra}
\usepackage[colorlinks=true,allcolors=blue]{hyperref}
\usepackage{color}
\usepackage{soul}
\usepackage[dvipsnames]{xcolor}

\usepackage{ulem}

\def\phi{\varphi}
\def\epsilon{\varepsilon}

\def\Id{\textrm 1\!\!\!\!1}

\newcommand{\Mc}[1]{\mathcal{#1}}

\newcommand{\setR}{\mathbb{R}}

\newcommand{\setC}{\mathbb{C}}

\newcommand{\bra}[1]{\langle #1 |}
\newcommand{\ket}[1]{| #1 \rangle }

\newcommand{\ii}{i}
\newcommand{\A}{\hat a}
\newcommand{\ad}{\hat a^\dagger}

\everymath{\displaystyle}

\newtheorem{theorem}{Theorem}

\begin{document}

\title{Emergent gravity from the correlation of spin-$\tfrac{1}{2}$ systems coupled with a scalar field}





\author{Quentin Ansel}
\email{quentin.ansel@univ-fcomte.fr}
\affiliation{Institut UTINAM, CNRS UMR 6213, Universit\'{e} Bourgogne Franche-Comt\'{e}, Observatoire des Sciences de l'Univers THETA, 41 bis avenue de l'Observatoire, F-25010 Besan\c{c}on, France}
\orcid{0000-0002-4594-5978}

\date{\today}

\begin{abstract}
This paper introduces several ideas of emergent gravity, which come from a system similar to an ensemble of quantum spin-$\tfrac{1}{2}$ particles. To derive a physically relevant theory, the model is constructed by quantizing a scalar field in curved space-time. The quantization is based on a classical discretization of the system, but contrary to famous approaches, like loop quantum gravity or causal triangulation, a Monte-Carlo based approach is used instead of a simplicial approximation of the space-time manifold. This avoids conceptual issues related to the choice of the lattice. Moreover, this allows us to easily encode the geometric structures of space, given by the geodesic length between points, into the mean value of a correlation operator between two spin-like systems. Numerical investigations show the relevance of the approach, and the presence of two regimes: a classical and a quantum regime. The latter is obtained when the density of points reaches a given threshold. Finally, a multi-scale analysis is given, where the classical model is recovered from the full quantum one. Each step of the classical limit is illustrated with numerical computations, showing the very good convergence towards the classical limit and the computational efficiency of the theory.
\end{abstract}

\maketitle
\tableofcontents
\section{Introduction}

Emergent gravity is a field of growing interest. The idea that the curved space-time of general relativity is not fundamental but comes from an effective theory has been explored in plenty of different approaches, such as in string theory~\cite{klammer2008fermions,steinacker2010emergent,Viennot_2021,viennot_fuzzy_2022}, quantum graphity~\cite{konopka_quantum_2008,caravelli_properties_2011,quach_domain_2012}, the work of Carroll and a.l.~\cite{cao_space_2017,bao_hilbert_2017,PhysRevD.97.086003} and many other~\cite{padmanabhan2015emergent,andrea_mondino_optimal_2022,gorard_quantum_2020,gorard_relativistic_2020}. Each approach is usually well motivated, and it shields light on some key elements of what could be a complete theory of quantum gravity. Unfortunately, these approaches are not fully satisfactory yet~\cite{padmanabhan2015emergent}. If the lack of experimental evidence is left aside,the current issues come partially from computational problems: from the main equation it becomes hardly possible to make physical predictions that can be tested~\cite{universe5010035}, or it is arduous to connect altogether the main ingredients expected in a complete theory of quantum gravity~\cite{ashtekar2021short,yosifov2021aspects,physics5010001}. In this paper, I propose yet another step forward in the direction of a complete theory of emergent quantum gravity.

 The ideas presented here do not cast exactly in the main theories, but it can be seen at the crossing between several of them. One of the main sources of inspiration is the work of Carroll and a.l.~\cite{cao_space_2017,bao_hilbert_2017,PhysRevD.97.086003}  In their series of papers, they develop the idea that the Einstein field equation is encoded in some way by the correlation between parts of the wave function describing the entire universe. They postulate that the correlation, quantified by the quantum mutual information, gives us areas between adjacent region, and they are able to recover, at least partially, Einstein field equations. The relation between the mutual information and space-time areas is strongly motivated by different arguments, the main ones being that the entropy of the horizon is proportional to its area, and the mutual information follows an approximated area law~\cite{PhysRevLett.100.070502}. Unfortunately, the inverse radon transform used to recover the metric from the data of all areas is far from being simple and not well understood in dimensions larger than 2. Moreover, the coupling with matter fields is not straightforward, and it seems difficult to recover quantum field theory in curved space-time with an adiabatic elimination~\cite{Viennot_2021}. Nonetheless, the idea that geometric quantities are encoded in the correlation between subparts of a quantum state is interesting, and the fundamental discrete nature of space-time follows other approaches, such as causal-networks~\cite{ambjorn_causal_2013}, loop quantum gravity~\cite{rovelli_quantum_2004,rovelli_covariant_2014,ashtekar2021short}, and quantum graphity~\cite{konopka_quantum_2008,caravelli_properties_2011,quach_domain_2012}. In these approaches, space-time is given by "atoms of space" connected together in a lattice. One drawback of these approaches is that the connectivity between these atoms of space is not imposed by fundamental principles, and one usually assumes 4 valent nodes (discretization with tetrahedrons) or 6-valent nodes (discretization with cuboids), but this is only motivated by a computational convenience. This discretization can be thought of as a quantum gravity analog of the nearest neighbors approximation of solid-state physique, but the ultimate theory must be free of the choice of a lattice.
 
 Based on these observations, I propose a model of emergent gravity for which the main ingredients are the following: (i) the structure of space-time is encoded in the correlation of subparts of a quantum state, each subpart being tough as an "atom of space", a quantum region of space (ii) Correlations are related to the geodesic distance between two quantum regions, not the area (iii) Quantum field theory in curved space-time is recovered in the limit when gravitational degrees of freedom are eliminated adiabatically~\cite{Viennot_2021,viennot_fuzzy_2022} and when the size of the wave-packet is large compared to the Planck scale, and finally (iv) the quantum state describing the quantum space-time is similar to an ensemble of spin-$\tfrac{1}{2}$, a quantum region being described by a spin, and geometric quantities are given by a spin-spin correlation. The use of an ensemble of two-level systems is interesting because it limits the size of the Hilbert space, but we can imagine rewriting the theory with arbitrary spin-$j$ or even with harmonic oscillators. 
 
 This paper is organized as follows. In sec.~\ref{sec:Monte_Carlo_approx}, the Monte-Carlo discretization procedure is introduced. In Sec.~\ref{sec:quantum_scalar_field}, the scalar field is introduced and it is quantized using the material developed in Sec.~\ref{sec:Monte_Carlo_approx}. In Sec.~\ref{sec:Quantum space-time}, a quantum description of the gravitational field is developed so that quantum field theory in curved space-time is recovered in the appropriate limit. This part of the theory is still significantly incomplete, but the minimum is made to provide a well-defined theory for which explicit calculations en physical predictions can be made. In Sec.~\ref{sec:Multi-scales analysis}, a multi-scale analysis is provided, where the classical scalar field theory is recovered from the full quantum theory by means of successive approximations. Each approximation step is illustrated with numerical simulations. A conclusion and prospective views are given in Sec.~\ref{sec:conclusion}.

\section{From continuous to discontinuous field theory with a Monte-Carlo approximation}
\label{sec:Monte_Carlo_approx}
 
 As outlined in the introduction, space-time is assumed fundamentally discrete, and the continuous theory is obtained at a large scale, without assumptions on a lattice structure at the Planck scale. Here, quantum regions are not connected with the nearest neighbors, they are all connected together with more or less weight. This is quite different from the typical situation encountered in lattice quantum field theory (such as in Wilson’s approach of path integrals~\cite{christ_weights_1982,ren_matter_1988,teixeira_random_2013,
 yamamoto_lattice_2014,Brower:2016moq}, loop quantum gravity~\cite{rovelli_covariant_2014},...). The underlying idea is illustrated in Fig.~\ref{fig:discretization_idea}. Here, a probabilistic approach is used, and integrals are discretized with a Monte-Carlo method~\cite{caflisch_monte_1998}. 
 
\begin{figure*}
\begin{center}
a)
\includegraphics[width=0.3\textwidth]{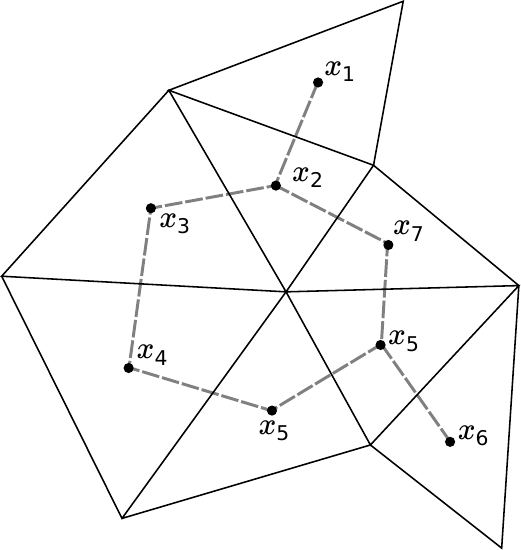}
b)
\includegraphics[width=0.3\textwidth]{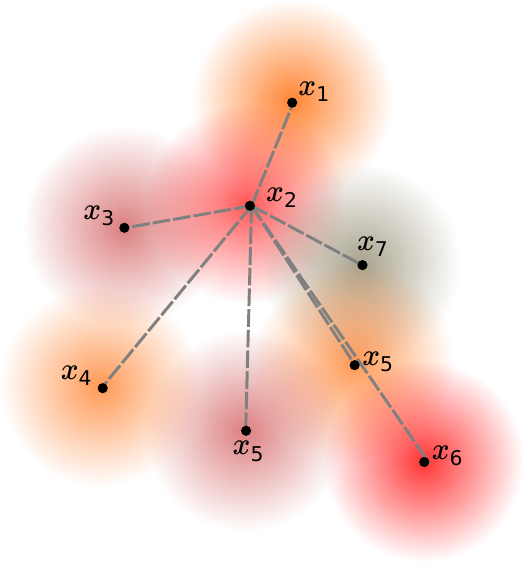}
\end{center}
\caption{a) A triangulated space. b) A discretized space with the vertices connected all together (only the edges between $x_2$ and the other vertices are shown) the weight between the vertices is a function of the distance between the points. In the continuum limit, the weights define density probability functions, depicted with colored gradients.}
\label{fig:discretization_idea}
\end{figure*}
 
 In the following, the points defining space are assumed \textit{uniformly distributed} on a manifold $\Sigma$, such that for any function $f:\Sigma \rightarrow \setC$, and any $R \subset \Sigma$ we have~\cite{zappa2018monte,de2018quasi}
 \begin{equation}
 \int_{R \in \Sigma} d^3x \sqrt{q(x)} f(x) = \lim_{N \rightarrow \infty} \frac{V_R}{N} \sum_{n=1}^N f(x_n)
 \end{equation}
where $V_R$ is the volume of $R$, and $q$ is the metric on $\Sigma$. In this paper, the Lebesgue measure is replaced by a probability distribution $p_{x}(y)$. Using the fact that $ \int_\Sigma d^3y \sqrt{q(y)} p_x(y)=1$, the normalization factor for an integration around the point $x$ is deduced to be 
\begin{equation}
\frac{V_x}{N}= \left(\sum_{n=1}^N p_{x}(y_n)\right)^{-1}.
\label{eq:estimation_of_the_volume}
\end{equation}
This result is written for a distribution with support on $\Sigma$, but it can be easily restricted to $R \subset \Sigma$.
This leads to a simple formula for the calculation of integral without reference to a coordinate system:
 \begin{equation}
 \begin{split}
\mathbb{E}[f]_x& = \int_{\Sigma} d^3y \sqrt{q(y)} p_x(y) f(y) \\
&\approx \frac{\sum_{n=1}^N f(y_n) p_x(y_n)}{\sum_{n=1}^N p_{x}(y_n)}.
\end{split}
 \end{equation}
To be true, $p_x(y)$ must be gauge invariant. A natural definition for such a quantity is based on the geodesic length between $x$ and $y$, which can be computed without too much difficulty as soon as the points are not too far from each other. In the following, it is assumed that
\begin{equation}
p_x(y) = \frac{\Mc N(R)}{(\sqrt{2 \pi} \epsilon)^3} \exp\left(- \frac{\sigma(x,y)}{\epsilon^2}\right)
\label{eq:assumption_probability_density}
\end{equation}
where $\epsilon$ has the dimension of a length, and $\sigma$ is Synge's world function~\cite{poisson_motion_2011}, which is equal to half the square of the geodesic length, and $\Mc N(R)$ is a normalization coefficient that depends on the space-time curvature. It is equal to 1 in a flat space.
In this setting, $\epsilon$ provides a natural notion of distance. It is not fixed by first principles, but it is assumed of the order of the Planck length, i.e., $\epsilon^2 \sim \hbar G/c^3 $. At a first sight, it is not obvious that Eq.~\eqref{eq:assumption_probability_density} gives a valid probability distribution in a curved space-time. This aspect is explored in appendix~\ref{sec:Normal law in a curved space-time} with a taylor expansion  of Synge's world function in Riemann normal coordinates.


%
The discretization procedure described here above can be easily generalized to the case when the distribution of points is not homogeneously distributed on $\Sigma$, but is distributed homogeneously with respect to a given coordinate system. In such a case, the previous formulas can be adapted using the replacement rule $p_x(y)\rightarrow \sqrt{q(y)} p_x(y)$.

\section{Discrete quantum scalar field theory in curved space-time}

\label{sec:quantum_scalar_field}

\subsection{Hamiltonian of a scalar field} 
Now that the general setting for the discretization of a field theory is made, a specific example can be studied. As a toy model, a real scalar field $\phi$ is considered. The starting point of the presentation is the Lagrangian density~\cite{khavkine_algebraic_2015,gerard_introduction_2018}:
 \begin{equation}
 \Mc L=\frac{\sqrt{-g}}{2} \left[ g^{\mu\nu} (\partial_\mu \phi)( \partial_\nu \phi) + m^2 \phi^2 \right]
 \label{eq:Lagrangian_scalar_field}
 \end{equation}
 where $g_{\mu \nu}$ is the metric tensor of the space-time manifold $\Mc M$, $\sqrt{-g} \equiv \sqrt{- \det g}$, and $m$ is the mass. Both $g$ and $\phi$ depend on the space-time position $x$. In the following, a 3+1 splitting is performed, to define a Hamiltonian~\cite{misner_gravitation_1973,rovelli_quantum_2004}. With this splitting, the metric is rewritten as follows: $g_{00}= -\mathsf{N}  + \mathsf{N}^a \mathsf{N}_a$, $g_{0a}=\mathsf{N}_a$, and $g_{ab} = q_{ab}$, where coordinate indices with Latin characters run over space coordinates only (i.e., $a,b=1,2,3$), and $q$ is the metric of a space sub-manifold $\Sigma \subset \Mc M$. With this notation, the volume density is $\sqrt{-g}=\sqrt{q}\mathsf{N}$. The next step is to determine the conjugated momentum of $\phi$. Using $\Pi = \tfrac{\partial \Mc L}{\partial (\partial_0 \phi)}$, one can determine
\begin{equation}
\Pi = \sqrt{q}\mathsf{N} g^{0 \nu} ~\partial_\nu \phi,
\end{equation}
and the Hamiltonian reads
\begin{equation}
H =\frac{1}{2} \int_\Sigma d^3 x~ \left(\frac{\Pi ^2}{\sqrt{q}\mathsf{N}} + \sqrt{q}\mathsf{N}\left[  q^{ab} \partial_a \phi \partial_b \phi + m^2 \phi^2 \right]\right).
\end{equation}
The Klein-Gordon equation in curved space-time is then recovered using Hamilton’s equations $\tfrac{\delta H}{\delta \phi} = - \partial_0 \Pi$ and $\tfrac{\delta H}{\delta \Pi} =  \partial^0 \phi$. To simplify the analysis, the coordinate gauge choice $\mathsf{N}=1$ and $\mathsf{N}_a=0$  is made. This avoids the use of many subtleties of algebraic QFT in curved space-time~\cite{khavkine_algebraic_2015,gerard_introduction_2018}, and thus, the derivation of a quantum theory is significantly simplified. A first look at a covariant approach to the theory is presented in Appendix~\ref{sec:covariant_theory}. The natural next step is to introduce complexifed field variables. They are defined by
 \begin{align}
\label{eq:def_a}
a(x)=\sqrt{\frac{m  }{2}}\left(\phi + \frac{\ii}{m \sqrt{q}}\Pi \right) \\
\label{eq:def_a_dagger}
a^\dagger (x)=\sqrt{\frac{m }{2}}\left(\phi - \frac{\ii}{m \sqrt{q}} \Pi \right).
\end{align}
With these new variables, the Hamiltonian takes the form
\begin{equation}
\begin{split}
H =&\int_\Sigma d^3 x \sqrt{q}~ \left(m\left(a^\dagger a + \frac{1}{2}\right)\right. \\
&\left. + \frac{1}{4m^2}  q^{ab} \partial_a (a^\dagger +a) \partial_b (a^\dagger +a) \right).
\end{split}
\label{eq:hamiltonian_complex_v1}
\end{equation}
Contrary to the standard approach, the Hamiltonian is not diagonalized by the complexified variables~\cite{itzykson_quantum_2012}. This enables the development of a quantum theory where the gravitational field plays a role only in the coupling terms between different space-time positions (the term with derivatives). 

\subsection{Discretization procedure}
\label{sec:discretization_scalar_field_Hamiltonian}

So far, the theory is given by a continuous space-time background. This section is devoted to its discretization. The most subtle point concerns the discretization of the coupling Hamiltonian $\partial^a \phi \partial_a \phi$, which must return a term of the form $\nabla_a \nabla^a \phi$ in the equation of motion, $\nabla$ being the covariant derivative  \footnote{In a curved space-time, the covariant derivative $\nabla$ replaces the usual derivative $\partial$, to take into account the modification of the tangent space $T_x\Mc M$ when a vector is moved in the manifold. For a scalar field, $\nabla$ reduces to $\partial$, but $\partial_a \phi = \nabla_a \phi$ is a vector and its derivative must be computed with $\nabla$~\cite{misner_gravitation_1973}.}. To apply the procedure of Sec.~\ref{sec:Monte_Carlo_approx}, Stokes's theorem~\cite{misner_gravitation_1973} must be used and a field vanishing on the boundary of the domain must be assumed. Therefore,
\begin{equation}
\int d^3x~ \sqrt{q}\mathsf{N}~  \partial_a \phi \partial^a \phi \rightarrow - \int d^3x~ \sqrt{q}\mathsf{N} ~ \phi~  \nabla_a \nabla^a \phi.
\end{equation}
with this setting, the Laplacian is related to the average of the field around a given position. More precisely, one has
 \begin{equation}
 \begin{split}
 \mathbb{E}[ \phi ]_x & = \int_\Sigma d^3y \sqrt{q}~ p_x(y) \phi(y) \\
 & =\phi(x) + \int_\Sigma d^3y \sqrt{q}~ p_x(y) ~y^a \partial_a \phi(x) \\
 &~ ~ ~ ~~  + \frac{1}{2} \int_\Sigma d^3y \sqrt{q}~ p_x(y)~  y^a y^b \nabla_a \nabla_b \phi (x) + ...  \\
 & = \phi(x) + \frac{\epsilon^2}{2}\nabla_a \nabla^a \phi(x) + O(\epsilon^3),
\end{split}
\label{eq:taylor_expand_mean_phi_square}
 \end{equation}
and as a consequence,
 \begin{equation}
 \nabla_a \nabla^a \phi(x) = \frac{2}{ \epsilon^2}\left[  \mathbb{E}[ \phi ]_x - \phi(x) \right] + O(\epsilon^3).
 \label{eq:secon_order_approx_kinetic_energy_density}
 \end{equation}
This result is easily obtained in a flat space, but a few more cautions must be taken in curved space. The validity of the $O(\epsilon^3)$ approximation is explored in appendix~\ref{sec:Normal law in a curved space-time}.

Next, the last equation is plugged into the integral over $\Sigma$, and a Monte-Carlo approximation is made. The result is
\begin{equation}
\int_\Sigma d^3y \sqrt{q}\mathsf{N}  \phi  \nabla_a \nabla^a \phi \approx \frac{2 V }{\epsilon^2 N}\sum_{n=1}^N \mathsf{N}_n \left[\phi_n \mathbb{E}[ \phi ]_n  - \phi^2_n\right].
\label{eq:first_order_approx}
\end{equation}
The subscript $n$ is used to specify that quantities are evaluated at the point $x_n$.

Complexified variables can be restored and Eq.~\eqref{eq:first_order_approx} can be inserted into the Hamiltonian~\eqref{eq:hamiltonian_complex_v1}. The theory is quantized straightforwardly by promoting the complexified field variables into creation and annihilation operators acting on a Fock space. The quantum Hamiltonian is

\begin{equation}
\begin{split}
\hat H =&\frac{V }{N} \sum_{n=1}^N  \left(m\left(\ad_n \A_n + \frac{1}{2}\right)\right. \\
&\left. - \frac{1}{2 \epsilon^2 m } \left[ (\ad_n + \A_n) \mathbb{E}[( \ad + \A)]_{x_n}  - (\ad_n + \A_n)^2\right] \right),
\end{split}
\label{eq:hamiltonian_quantum}
\end{equation}
 with $\ad_n$ and $\A_n$ the creation/annihilation operators at the position $x_n$. Despite the relatively simple expression of the Hamiltonian, the average $\mathbb{E}$ hides a factor $V/N$, which plays the role of an unknown coupling constant between oscillators. To circumvent this issue, the Hamiltonian is factorized with a global factor $V/N$, and at each position $x_n$ the volume is evaluated with the relation~\eqref{eq:estimation_of_the_volume}. The global factor is only responsible for a rescale of the time coordinate in the Schrödinger equation. The final result is:
\begin{equation}
\begin{split}
\hat H =&\frac{V^2 }{N^2} \sum_{n=1}^N   \left( \sum_{k=1}^n p_{x_n}(x_k)\right)\left(m\left(\ad_n \A_n + \frac{1}{2}\right)  \right. \\
& \left. + \frac{1}{ 2 \epsilon^2 m }(\ad_n + \A_n)^2 \right) \\
& - \frac{1}{2 \epsilon^2 m } \left( (\ad_n + \A_n) \sum_{k=1}^N p_{x_n}(x_k )( \ad_k + \A_k)  \right) .
\end{split}
\label{eq:hamiltonian_quantum_final}
\end{equation}
This Hamiltonian shares very strong similarities with the Bose-Hubbard model~\cite{PhysRevB.40.546,dutta2015non,mivehvar2021cavity,arovas2022hubbard,
aidelsburger2022cold}. In particular, this model was used to simulate quantum fields in curved space-times~\cite{BEMANI2018186}. While the standard approach is to start from a given lattice configuration of the Bose-Hubbard model and to compute the resulting effective space-time, here the reverse approach is made. However, Eq.~\eqref{eq:hamiltonian_quantum_final} is free from the nearest neighbor approximation, and it provides a simple mapping between the space-time geometry and the parameters of the Hamiltonian.

\section{Quantum space-time}
\label{sec:Quantum space-time}

In the previous sections, a formalism that allows us to discretize a quantum scalar field in space was developed, where gravitational degrees of freedom have an influence only through an interaction Hamiltonian that couples the fields at different locations. With these materials at hand, a (partial) quantum theory for the gravitational field consistent with QFT in curved space-time can be constructed.

A closer look at Eq.~\eqref{eq:hamiltonian_quantum_final} highlights that in the interaction Hamiltonian, the geometry of $\Mc M$ is hidden in $P_{x_n}(x_k)$. These weights are the quantities that must be quantized.

The quantization relies on several assumptions. In particular, the geometry of the gravitational state is encoded in the correlation between different sites of an ensemble of spin-$\tfrac{1}{2}$, a spin being associated with a point $x_n$. Then, a general quantum state $\ket{\Psi_G}$ is a linear combination of vectors of $\Mc H_G =\setC^{2N}$.

Many different operators can be used to describe a two-point correlation function. Here, the following choice is made:
\begin{equation}
\hat C_{nk} = \hat \sigma_+(x_n) \hat \sigma_-(x_k) +\hat \sigma_+(x_k) \hat \sigma_-(x_n) 
\label{eq:def_operator_Cnk}
\end{equation}
with $\hat \sigma_-(x_n)$ and $\hat \sigma_+(x_n)$ the spin creation/annihilation operators associated with the point $x_n$.
The mean value of this operator, $C_{nk} =\bra{\Psi_G} \hat C_{nk} \ket{\Psi_G}$, gives us information on how much two regions are correlated, and more specifically how much two regions are entangled since a superposition of different states is necessary to get $C_{nk} \neq 0$. 

To relate correlations and geometric data of the manifold, I propose to postulate the following relation between the correlation and the geodesic distance:
\begin{equation}
C_{nk}^2 = \alpha ~ \exp \left(- \frac{\sigma(x_n,x_k)}{\epsilon^2}\right),
\label{eq:def_coef_Cnk}
\end{equation}
with $\alpha$ a proportionality coefficient. This definition, motivated by a numerical investigations, requires the exclusion of the terms $n=k$ in the sum. This is because $\hat \sigma_+(x_n) \hat \sigma_-(x_n) +\hat \sigma_+(x_n) \hat \sigma_-(x_n) = \Id_2$, and thus, by taking the expectation value, it is not possible to recover the coefficient $\alpha$. As a consequence, the condition $p_{x_n}(x_n) = 0$ must be imposed in Eq.~\eqref{eq:hamiltonian_quantum_final}. This is not a problem for the convergence of the Monte-Carlo approximation scheme since it concerns a modification of the probability distribution whose Lebesgue measure is equal to zero.

Eq.~\eqref{eq:def_coef_Cnk} has a very simple physical interpretation. The larger the distance between $x_n$ and $x_k$, the smaller the correlation. In order to suppress the quantum fuzziness at large scales, the correlation must decay sufficiently fast. A normal law (see Eq.~\eqref{eq:assumption_probability_density}) is assumed for convenience, but it is not excluded that another distribution can be more adapted to the problem. For example, $C_{nk}^2 \propto \exp(- \sqrt{2 \sigma_{nk}}/\epsilon)$ could be also a viable choice. The validity of this ansazt is explored numerically in Fig.~\ref{fig:convergence_approximation_scheme}. The approach of the numerical investigation is the following. A cost function $\Mc F = \frac{1}{N}\sum_{n=1}^{N}\sum_{n=1}^{N} | \alpha e^{- \sigma_{nk}/\epsilon^2} - C_{nk}^2(\Psi_G)|^2$ is minimized with respect to the vector components of $\ket{\Psi_G}$. For a different number of points (i.e., a different number of spins), a quantum state is optimized numerically such that $C_{nk}^2$ is as close as possible to $\alpha e^{- \sigma_{nk}/\epsilon^2}$. The value of $\alpha$ is fixed by hand, to simplify the optimization process. Specific details on these numerical computations are gathered in Appendix~\ref{sec:Details on the numerical optimizations}. Different kinds of configurations and spaces are investigated, with flat and curved spaces, and with points organized randomly or regularly. In all situations, very low values of $\Mc F$ are obtained, except when the density of points is too large. For example, in the case of Fig.~\ref{fig:convergence_approximation_scheme} b), the value of $\Mc F$ drops significantly when $L$ reaches a given threshold ($L/\epsilon \approx 1.5$ for the cube and $L/\epsilon \approx 1.2$ for the 3-simplex). The threshold depends on the number of spins and it seems to correspond to $C_{nk} \approx 1/\sqrt{N}$. Further investigations may clarify this point, but this numerical result has an important physical consequence: \textit{there is a natural minimum scale for which we can obtain a classical geometry}. The value of the threshold depends on $\alpha$, and thus it is not clear at this point in which exact circumstances it is possible to encounter a quantum space-time. This issue will be investigated elsewhere.

\begin{figure}
\includegraphics[width=\columnwidth]{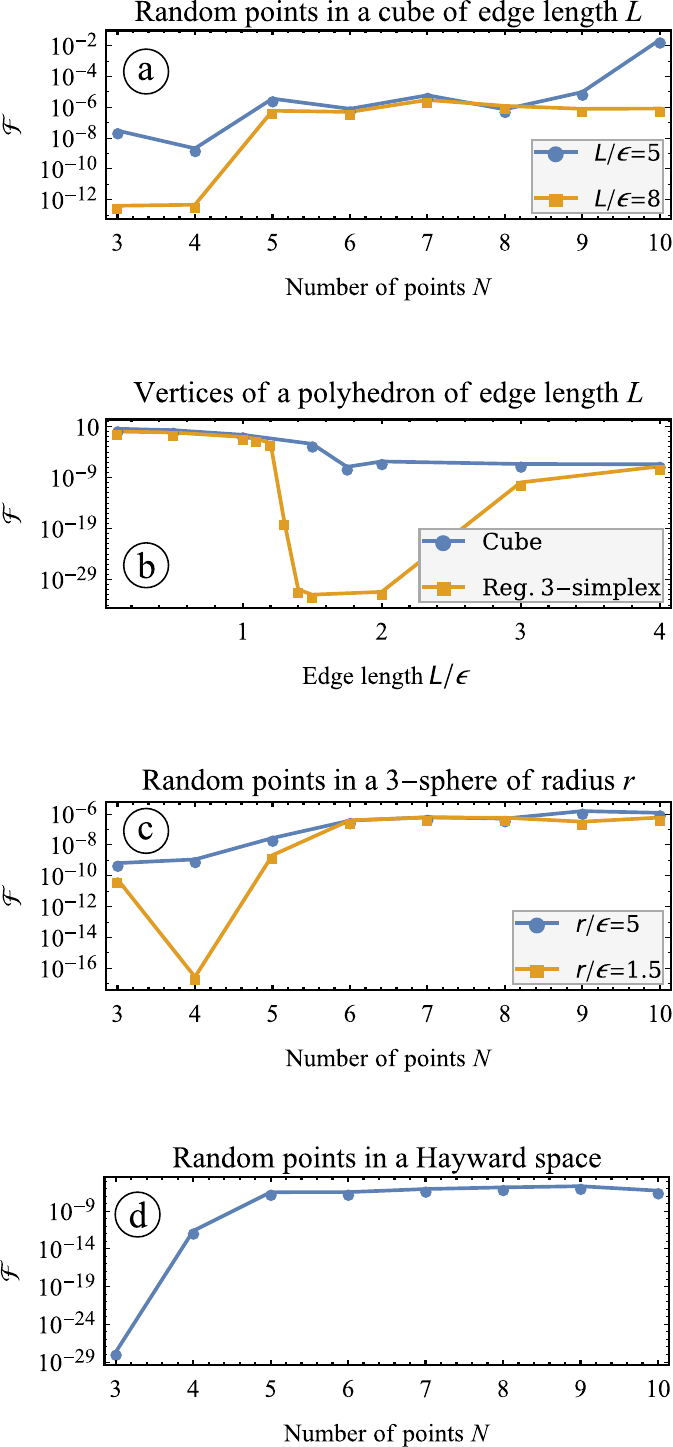}
\caption{Test of the validity of the gravitational wave function ansazt, using the numerical minimization of the cost function $\Mc F = \frac{1}{N}\sum_{n=1}^{N}\sum_{n=1}^{N} |\alpha \exp(- \sigma(x_n,x_k)/\epsilon^2) - C_{nk}^2(\Psi_G)|^2$. The parameters to optimize are vector components of $\ket{\Psi_G}$. Different case studies are considered. In the first one (panel a), the locations of $N$ points are generated randomly in a cube of edge length $L/\epsilon$. In the second case (panel b), the number of points is fixed, and they are located at the vertices of a regular polyhedron (a cube or a 3-simplex), the edge length being varied. In the third case (panel c), the situation is similar to the first one, but the cube is replaced by a 3-sphere of radius $r$. In the last case (panel d), points are randomly sampled in a Hayward space-time~\cite{PhysRevLett.96.031103}. Further details concerning the numerical calculations are given in Appendix~\ref{sec:Details on the numerical optimizations}.}
\label{fig:convergence_approximation_scheme}
\end{figure}

The Hamiltonian of the scalar field \eqref{eq:hamiltonian_quantum_final} can be reconstructed with the information given by $C_{nk}$. However, the square of $C_{nk}$ is required, to ensure a positive quantity that can be assimilated with a probability. This cannot be realized as the mean value of a linear operator acting on $\ket{\Psi_G}$, the quadraticity can be recovered with a duplicate of the state and a tensor product. We introduce $\ket{\tilde \Psi_G} = \ket{\Psi_G} \otimes \ket{\Psi_G}$, and with this state, $C_{nk}^2$ is obtained with the relation
\begin{equation}
\bra{\tilde \Psi_G} \hat C_{nk} \otimes \hat C_{nk}\ket{\tilde \Psi_G} = C_{nk}^2.
\end{equation}

The full quantum Hamiltonian is obtained by replacing $P_{x_n}(x_k)$ by $\hat C_{nk} \otimes \hat C_{nk}$ in Eq.~\eqref{eq:hamiltonian_quantum_final}. The parameter $\alpha$ is omitted since it can be factorized with the term $V^2/N^2$, which is already a global scaling factor of the Hamiltonian.

The dynamics of the full quantum system are given by the Schrödinger equation
\begin{equation}
\frac{d }{dt} \ket{\Psi} = - \frac{\ii N}{V} \left( \hat H + \hat H_G\right) \ket{\Psi},
\end{equation}
where $\hat H$ is the Hamiltonian~\cite{misner_gravitation_1973,rovelli_quantum_2004} given in Eq.~\eqref{eq:hamiltonian_quantum_final}, and the origin of factor $N/V$ is explained in appendix~\ref{sec:classical_hamiltonian}. $\hat H_G$ is a gravitational Hamiltonian. Deriving the explicit expression of $\hat H_G$ is a nontrivial key point of the theory. It is left unresolved in this article. Both $\hat H$ and $\hat H_G$ act on gravitational degrees of freedoms, but in the low coupling limit, $\ket{\Psi}$ approximately take the form of a tensor product at any time, i.e. $\ket{\Psi}\approx\ket{\tilde{ \Psi}_G}\otimes \ket{\Psi_{S.F.}} $ (Born approximation). This limit is not obtained by $m \rightarrow 0$, but when $\Vert \hat H \ket{\Psi} \Vert \ll \Vert \hat H_G \ket{\Psi} \Vert$. By taking the trace over gravitational degrees of freedom, one gets:
\begin{equation}
\frac{d }{dt} \ket{\Psi_{S.F.}} = - \frac{\ii N}{V} \left( \langle \hat H \rangle_{\tilde \Psi_G} + \langle\hat H_G \rangle_{\tilde \Psi_G} \Id \right) \ket{\Psi_{S.F.}}
\end{equation}
The mean value of $\hat H_G $ returns only an irrelevant phase from the point of view of the scalar field, and thus, the Schrödinger equation can be redefined as
\begin{equation}
\frac{d }{dt} \ket{\Psi_{S.F.}} = - \frac{\ii N}{V}  \langle \hat H \rangle_{\tilde \Psi_G}  \ket{\Psi_{S.F.}}.
\label{eq:effective_schordinger_scalar_field}
\end{equation}
It gives us the dynamics of the scalar field in the limit of quantum fields in curved space-time. This is a mean-field Schrödinger equation similar to the one encountered in quantum optics~\cite{peng_introduction_1998,gardiner_quantum_2004,breuer_theory_2007}. From the point of view of open quantum systems, this corresponds to a first-order adiabatic elimination~\cite{azouit_adiabatic_2016}. The second-order expansion introduces non-unitary dynamics and an effective Lindblad equation must replace Eq.~\eqref{eq:effective_schordinger_scalar_field}.  The second-order correction goes beyond the standard quantum field theory in curved space-time and it incorporates the notion of gravitational decoherence~\cite{Bassi_2017,viennot2017adiabatic}. The rest of this paper is focused on the first-order approximation and only Eq.~\eqref{eq:effective_schordinger_scalar_field} is considered.

\section{Multi-scales analysis}

\label{sec:Multi-scales analysis}
In the previous sections, a quantum theory for the scalar field and the gravitational field was developed, such that quantum field theory in curved space-time is recovered when suitable limits are taken. Several aspects have been omitted and the goal of this section is to go deeper in the analysis, using numerical calculations. 

\subsection{Classical and non-classical gravitational states}

\begin{figure*}[t!]
\includegraphics[width=\textwidth]{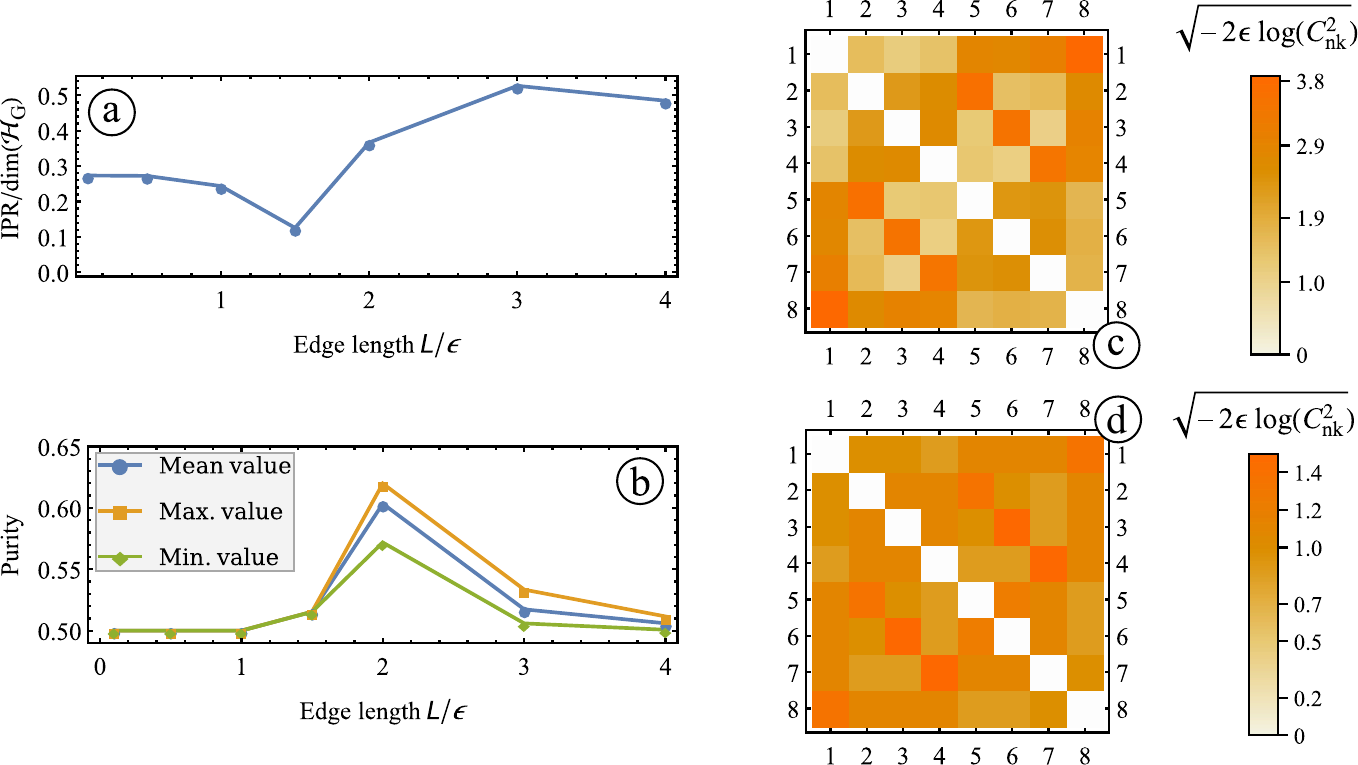}
\caption{Panel a) shows the IPR of $\ket{\Psi_G}$, given by $IPR = \left(\sum_{m=1}^{dim {\Mc H_G}} |\Psi_{G,m}|^4\right)^{-1}$ as a function of the edge length $L$ for the optimized quantum states whose final cost function are given in Fig.~\ref{fig:convergence_approximation_scheme} b). In this situation, the optimized quantum state aims to describe a classical geometry with 8 points located at the vertices of a cube. We recall that when $L/\epsilon\leq 1.5$, the state fails to describe the target geometry, but when $L/\epsilon >1.5$ the geometry is well described by the $\ket{\Psi_G}$. Panel b) is the same as panel a), but it shows the mean, max. and min. values of the purity $\text{Tr}[\hat \rho(x_n)^2]$ of each spin reduced states. Panel c) and d) show the geodesic distance between each point, as given respectively by the optimized state at $L/\epsilon=2$ and $L/\epsilon =0.1$. The geodesic distance is deduced from Eq.~\eqref{eq:def_coef_Cnk}, with value $\alpha = (\sqrt{2\pi} \epsilon)^3 $ (fixed for the numerical optimization), and it reads: $l_{nk}=\sqrt{-2\epsilon \log(C_{nk}^2)}$.}
\label{fig:gravity_quantum_state}
\end{figure*}

First, a few properties of the gravitational quantum state are investigated by computing a few relevant quantities. A first one is the inverse participation ratio (IPR)~\cite{kramer_localization_1993}, defined by $IPR = \left(\sum_{m=1}^{dim {\Mc H_G}} |\Psi_{G,m}|^4\right)^{-1}$, with $\Psi_{G,m}$ the vector components of $\ket{\Psi_{G}}$ in the canonical basis of $H_G$. Roughly, the IPR gives us the number of basis states on which the state is decomposed with equal weights. This provides us some information on the complexity of the quantum state, and how many basis states are used to encode the data of a classical space-time. A second interesting quantity is the purity of the reduced state associated with each spin. 

The purity associated with a  point $x_n$ is defined by $\text{Tr}[\hat \rho(x_n)^2]$~\cite{gardiner_quantum_2004,breuer_theory_2007} , with $\hat \rho(x_n)$ the reduced density matrix of the corresponding spin. Since the system is a spin-$\tfrac{1}{2}$, the reduced density matrix can be easily computed using $\hat \rho(x_n) = (\Id_2 + \sum_{i=1}^3 \hat \sigma_i \bra{\Psi_G} \hat \sigma_i(x_n) \ket{\Psi_G})/2$, with $\sigma_i$ the Pauli matrices. The purity is a quantity describing how the density matrix is close to a pure state. The purity of a pure state is equal to 1, and the minimum value is 0.5 for the state $\hat \rho = \Id_2 /2$.

The IPR and the averaged purity of each $x_n$ are given in Fig.~\ref{fig:gravity_quantum_state} a) and b), for the quantum states optimized to describe a cube, whose cost functions $\Mc F$ are given in Fig.~\ref{fig:convergence_approximation_scheme} b). With these states, there are the ones that describe a classical geometry, and the ones that do not. The physical relevance of this second kind of state is less important than the first one, because they fail in the description of a classical object, and there is no guarantee that they describe a true physical situation. However, they give us elements of comparison with the states describing a classical geometry. We observe that the IPR is usually quite high (between 30\% and 50\% of the total dimension of the Hilbert space) and the purity is low, close to the minimum value of 0.5. The local states are therefore highly statistical mixtures. The relatively high IPR and low purity suggest that the states are quite highly entangled. Interestingly, the data of Fig.~\ref{fig:gravity_quantum_state} a) and b) clearly indicates the change of behavior near $L/\epsilon = 1.5$. This observation, already commented in the previous section, requires a deeper analysis to explain its origin.

Having a large-scale entanglement and locally mixed states seem to be necessary conditions to obtain classical geometries, but we see that there exist non-classical states that have similar behaviors. Only the correlation function allows us to distinguish a classical to a non-classical geometry. In figure  Fig.~\ref{fig:gravity_quantum_state} c) and d) are plotted the geodesic distance between the $x_n$ for the optimized state computed with $L/\epsilon = 2$ and $L/\epsilon = 0.1$. They are respectively classical and non-classical states.  The geodesic distance is computed from $C_{nk}$ by inverting Eq.~\eqref{eq:def_coef_Cnk}, and using $\alpha = (\sqrt{2\pi} \epsilon)^3 $ (fixed for the numerical optimization). It reads $l_{nk}=\sqrt{-2\epsilon \log(C_{nk}^2)}$. In the classical case, the distances correspond to the ones between the vertices of a cube of edge length 2 (they are equal to $2$, $2\sqrt{2}\approx 2.8 $ and $2\sqrt{3}\approx 3.5$). In the nonclassical case, all the distances are $\approx 1.5$. This is not possible unless the geometric figure is a hypercube of dimension equal to 8, or if the interpretation of well-located points in a 3D space is lost. This ensemble could also be seen as a kind of fuzzy object. This object is divided into well-identifiable subsystems but do not have well defined positions in a 3D space. They are identically close to one another, violating the basis axiom of normed 3D spaces.

\subsection{Classical limit of the quantum scalar field}

\begin{figure*}[t!]
\includegraphics[width=\textwidth]{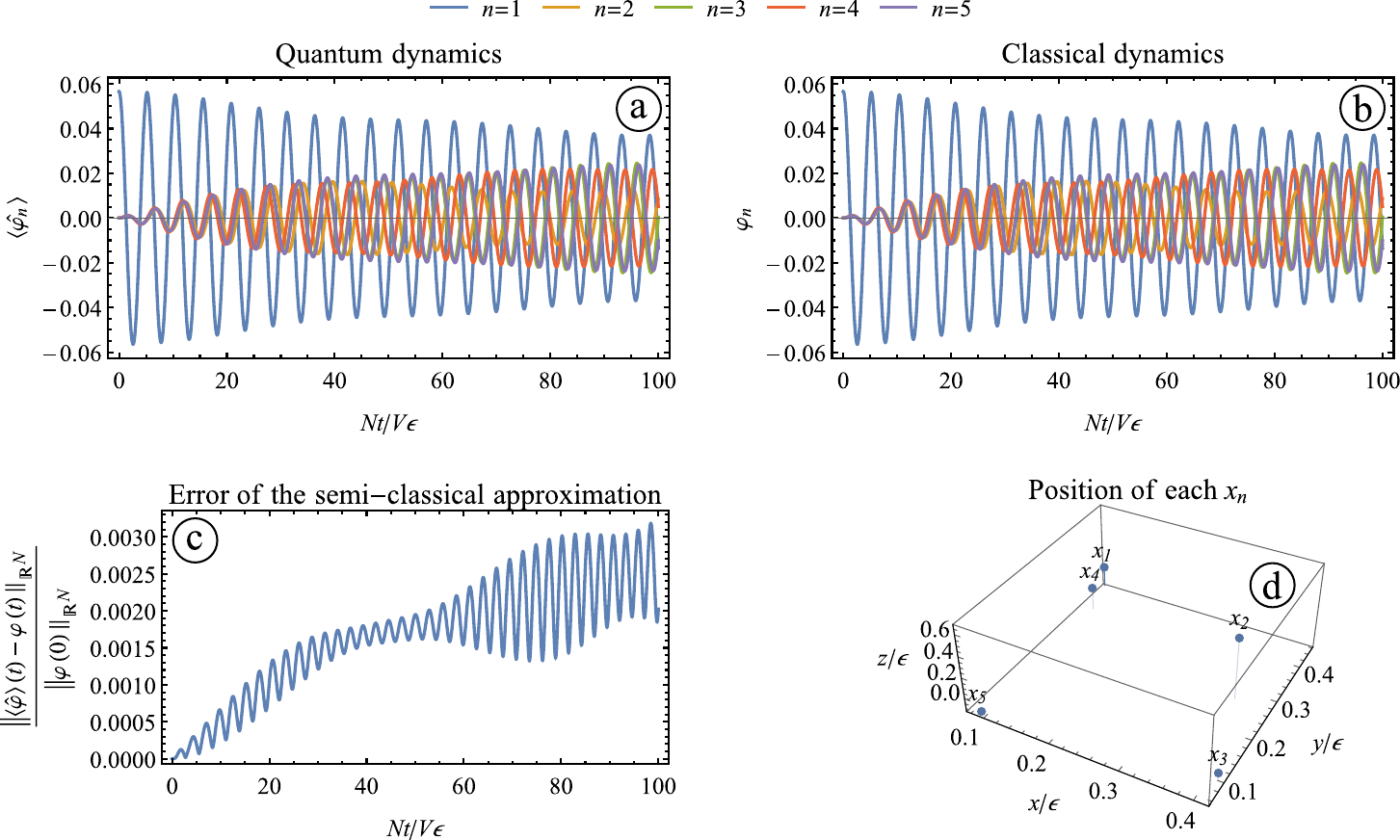}
\caption{Panels a) and b) show the time evolution of the field at different space positions, computed using the Schrödinger equation and the Hamiltonian Eq.~\eqref{eq:hamiltonian_quantum_final} for the quantum case, or the classical Hamiltonian and Hamilton equations for the classical case given in Appendix~\ref{sec:classical_hamiltonian}. The system is defined with $5$ points in Minkowski space whose locations are given in panel d). The initial quantum state condition is a coherent state $\ket{\Psi_{S.F.}}=\ket{0.1}\otimes \ket{0} \otimes \ket{0} \otimes \ket{0} \otimes \ket{0}$, and the initial classical state is $\phi_n =(\bra{0.1} \hat \phi \ket{0.1},0,0,0,0)$ and $\tilde \Pi_n = (0,0,0,0,0)$. The mass is fixed such that $m \epsilon=2$. Panel c) shows the error made by the semi-classical approximation. It consists of the 2-norm in $\setR^N$ of the difference between the vector $\langle \hat \phi_n \rangle$ and $\phi_n$. The accumulation of error is due to the small nonclassical effects, but also to some extent numerical errors due to the truncation of the Hilbert space (only the first four states of the Fock space are kept, leading to a total Hilbert space dimension of 1024). } 
\label{fig:quantum_vs_classical_scalar_field_v2}
\end{figure*}

The next step of the multi-scales analysis is to consider a regime where the gravitational field is in its semi-classical limit such that the dynamics of the scalar field are governed by the Schrödinger equation Eq.~\eqref{eq:effective_schordinger_scalar_field}. In this limit, the mean value of the Hamiltonian with respect to the gravitational state gives us the coefficients $p_{x_k}(x_n)$ in Eq.~\eqref{eq:hamiltonian_quantum_final}. We can thus forget the quantum nature of gravity and the coefficients $p_{x_k}(x_n)$ become inputs of the model. As previously outlined in the text, the Hamiltonian is in the form of a Bose-Hubbard Hamiltonian, and it is well defined and already studied in the literature. However, it is not so easy to recover standard quantum field theory, and even harder to recover the classical field. Of course, we can follow the reverse procedure of the one followed in Sec.~\ref{sec:quantum_scalar_field}, and come back to the initial point, but it is quite difficult to visualize the transition from the quantum world to the semi-classical one. To this end, different semi-classical approximations can be performed. Here, we are interested in an approach particularly well suited for numerical computations. First, we consider the regime where $\langle \hat \phi  \rangle \approx \phi$ and  $\langle \hat \Pi \rangle  \approx  \Pi$~\cite{gardiner_quantum_2004}, where the dynamics of $\phi$ and $\Pi$ are computed using Hamilton’s equations, with the classical Hamiltonian as a starting point (see appendix \ref{sec:classical_hamiltonian}). This allows us to replace the Schrödinger equation with a set of classical equations of motion, with a small number of degrees of freedom (compared to the size of the truncated Hilbert space used in a numerical simulation)~\cite{sachdev1999quantum,PhysRevA.76.032116,Veksler_2015}. Physically, the classical system corresponds to a set of coupled oscillators. Equations of motions are non-linear, but their numerical integration is not particularly difficult. Next, we take the limit when the number of discretization points tends to infinity, to recover the continuum limit. 

The regime with $\langle \hat \phi  \rangle \approx \phi$ and  $\langle \hat \Pi \rangle  \approx  \Pi$ is obtained when the quantum state is of the form
\begin{equation}
\ket{\Psi_{S.F.}} = \bigotimes_{n=1}^N \ket{a_n}
\end{equation}

with $\ket{a_n}$ a coherent state defined by~\cite{gardiner_quantum_2004}
\begin{equation}
\ket{a_n} = e^{-|a_n|^2/2} \sum_{k=0}^\infty \frac{(a_n)^k}{\sqrt{k!}}\ket{k_n},
\end{equation}
$\ket{k_n}$ being the Fock basis state of the oscillator $n$. A coherent state has the particularity that $\bra{a_n} \hat a_n \ket{a_n} = a_n$ and $\bra{a_n} \hat a_n^\dagger \ket{a_n} = a_n^*$, and thus $\langle \hat \phi  \rangle = \phi$ and  $\langle \hat \Pi \rangle =  \Pi$. The Hamiltonian has the particularity to preserve quite well the coherent state structure of the state, and thus, a coherent state remains coherent as a function of time. As a consequence, if $\ket{\Psi_{S.F.}}$ is a coherent state, it behaves like a classical ensemble of coupled oscillators. This point is illustrated in Fig.~\ref{fig:quantum_vs_classical_scalar_field_v2}. We see that quantum and classical dynamics are very close to each other, even after very long times (i.e., after a large number of oscillations). Physically, this system is quite far from a classical continuous scalar field because the number of points taken in this example is quite small. The number of quantum oscillators that can be considered in a numerical simulation is limited by the size of the Hilbert space. Nevertheless, a feature of the continuous limit can already be observed: an initially localized excitation will tend to propagate towards the nearby points. This is clearly visible in panels a) and b), where the field amplitude in $x_1$ decreases while the amplitude in the other points increases.

Next, the state of the scalar field is assumed to stay coherent at any time, such that dynamics can be computed using only the classical equation of motion (see appendix \ref{sec:classical_hamiltonian}). In this regime, the number of $x_n$ in the numerical simulation can be increased considerably, but the Monte Carlo method requires quite a lot of points to obtain a good convergence, and thus, the illustrative example is limited to a 1D flat space. In practice, the expression of the Hamiltonian remains the same as in 3D, but the coefficients $p_{x_n}(x_k)$ are now defined by. $p_{x_n}(x_k)=e^{-(x_n-x_k)^2/(2 \epsilon^2)}/(\sqrt{2\pi}\epsilon)$. In a flat 1D space, the Klein-Gordon equation is given by~\cite{itzykson_quantum_2012}
\begin{equation}
\frac{\partial^2 \phi}{\partial x^2} -\frac{\partial^2 \phi}{\partial t^2} - m^2 \phi = 0. 
\label{eq:Klein_Gordon_flat}
\end{equation}
It can be derived readily from the Lagrangian Eq.~\eqref{eq:Lagrangian_scalar_field}, with $g^{\mu\nu} = \eta^{\mu\nu}$ (the Minkowski metric), and using Euler Lagrange equations with respect to the field variable $\phi$. Solutions of Eq.~\eqref{eq:Klein_Gordon_flat} are given by~\cite{itzykson_quantum_2012}
\begin{equation}
\phi (x,t) = \int \frac{dk}{4\pi \omega(k)}\left( a(k)e^{\ii (-\omega(k) t + kx)} + C.C. \right),
\label{eq:solution_KG_equation}
\end{equation}
with $\omega(k) = \sqrt{k^2+m^2}$. In order to satisfy the condition that $\phi$ vanish on the boundary of the domain (remember that this is a necessary condition to derive the quantum Hamiltonian in the form of Eq.~\eqref{eq:hamiltonian_quantum_final}), the density of modes is chosen to be
\begin{equation}
a(k)= \frac{\omega(k)}{\kappa} \exp\left(-\frac{k^2}{2 \kappa^2}\right).
\label{eq:density_of_mode_scalar_field}
\end{equation}
Plugging Eq.~\eqref{eq:density_of_mode_scalar_field} into Eq.~\eqref{eq:solution_KG_equation} leads us to
\begin{equation}
\phi (x,t) = \int \frac{dk}{2\pi \kappa } \exp\left(-\frac{k^2}{2 \kappa^2}\right) \cos(kx - \omega(k) t).
\label{eq:analytic_expr_scalar_field}
\end{equation}
This integral can be easily evaluated numerically with very good precision. Even if a short numerical computation is necessary for its evaluation, Eq.~\eqref{eq:analytic_expr_scalar_field} is called below "the analytic solution" of the Klein-Gordon equation. In the limit when $m \rightarrow \infty$, a simple approximated expression can be derived, $\phi(x,t)\approx  e^{-x^2/2\kappa^2} \cos(m t)/\sqrt{2\pi}$. The shape of the wave packet is a Gaussian oscillating at a frequency $m$. 

\begin{figure*}[t]
\includegraphics[width=\textwidth]{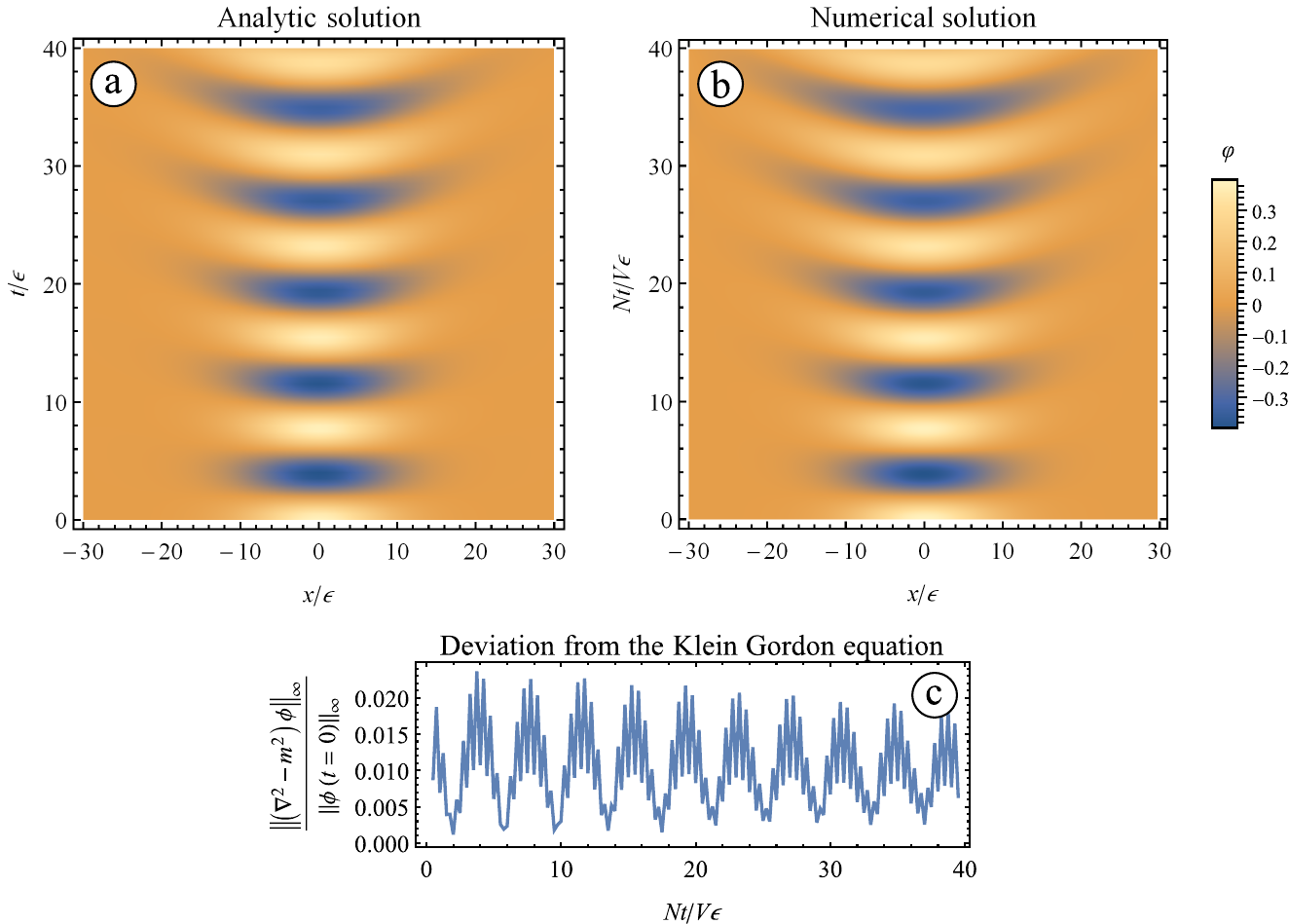}
\caption{Panel a) shows the analytic solution of the Klein Gordon equation, given by Eq.~\eqref{eq:analytic_expr_scalar_field} with $m \epsilon=0.8$ and $\kappa \epsilon = 0.15$. Panel b) shows the numerical solution, computed with the same parameters and an initial state given by the analytic solution at $t=0$. In the numerical solution, time must be rescaled due to the global factorization in the Hamiltonian Eq.~\eqref{eq:hamiltonian_quantum_final}. Panel c) gives the deviation (also called the error) of the numerical solution from the Klein-Gordon equation. It is given by the sup norm of $(\nabla^2 - m^2) \phi= \partial_x^2 \phi - \partial_t^2 \phi - m^2\phi$, the norm being computed over the variable $x$.  }
\label{fig:classical propagation}
\end{figure*}

In Fig.~\ref{fig:classical propagation}, the analytic solution of the Klein-Gordon equation is compared with the discrete model, defined by the classical version of Eq.~\eqref{eq:hamiltonian_quantum_final}. The equation of motion of the second model being entirely integrated numerically, the resulting solution is called below the "numerical solution". The two dynamics are very close to each other. In both cases, two observations can be made, the field oscillates at the frequency $m$ and there is a time deformation of the wave packet because $m \epsilon$ is of the order of unity. The validity of the numerical solution can be characterized more precisely by computing its deviation from the Klein-Gordon equation. The number of points $\phi(x,t)$ in the numerical solution being sufficiently high, it is possible to compute an analytic interpolation of the data points. The interpolated function is then inserted in the Klein-Gordon equation, and the error at a given time $t$ is defined by the sup norm, computed over the variable $x$. The error of the numerical solution is given in Fig.~\ref{fig:classical propagation} c). It is divided by the sup norm of $\phi(t=0)$ to fix an irrelevant scaling amplitude. The error remains small and bounded ($\lesssim$ 2 percent). Better accuracy of the method can be obtained by increasing the number of points, increasing the region of integration, and increasing the size of the wave packet (the size of the wave packet must be large to ensure a good convergence of the approximation scheme of the Laplacian, as given in Eq.~\eqref{eq:secon_order_approx_kinetic_energy_density}). 

To finalize the illustration of the method, a short example of propagation in a curved space-time is provided in Fig.~\ref{fig:Propagation_Hayward}. The chosen one is a Hayward space-time~\cite{PhysRevLett.96.031103} where space-slices are restricted to a single dimension. In Lemaitre type coordinates, the line element is given by:
\begin{equation}
ds^2 = -dt ^2 + \frac{r_s r^2}{r^3 + l^2 r_s} d\rho^2 
\end{equation}
with $t$ the time coordinate, and $\rho$ a space coordinate. $r_s$ correspond to the Schwarzschild radius, and $l$ is a length scale of the central body. The function $r(\rho,t)$ is a non-analytic, and defined through the differential equation $d\rho -dt = dr  \left(\tfrac{r_s r^2}{r^3 + l^2 r_s}\right)^{-1/2} $. It can be easily evaluated numerically or approximated analytically using a Taylor expansion. Since the coordinates are related by a differential equation, an arbitrary choice of origin is required. Here, $\rho-t=0$ is chosen to correspond with $r= 2^{1/3} l^{2/3} r_s^{1/3}$. In this case study, the space-time has two horizons~\cite{PhysRevLett.96.031103}. The wave packet is initially located near the inner horizon. Due to the non-zero energy density at the origin of this space-time, the wave packet propagates towards the horizon (free fall towards the center of the body, located at $r=0$), and a very strong distorsion happens near this latter. Here, the goal is just to illustrate that the method allows us modelize the non-trivial properties of a curved space-time\footnote{The second order expansion of $\sigma$ used for its estimation allows us to show the existance of curvature effects, but it may also leads to numerical errors that could be corrected with more accurate estimations of $\sigma$.}. No further details are investigated, this could be the topic of another independent paper.

The very good results given by the numerical solution show the relevance of the discretization method. It is, therefore, an interesting for quantum gravity purposes, but it can also be an alternative discretization scheme in quantum field theory in curved space-time for which most of the current techniques are based on a simplicial discretization of the space-time manifold~\cite{christ_weights_1982,ren_matter_1988,teixeira_random_2013,
 yamamoto_lattice_2014,Brower:2016moq}. The simplicial discretization requires a nontrivial discretization scheme of the differential operators which may be difficult to compute in large dimensions. With the Monte-Carlo based approach the nontrivial task is to compute the geodesic length, but once this task is achieved, potentially with some approximations, then, the implementation of the scheme is straightforward.

\begin{figure}[h]
\includegraphics[width=0.5\textwidth ]{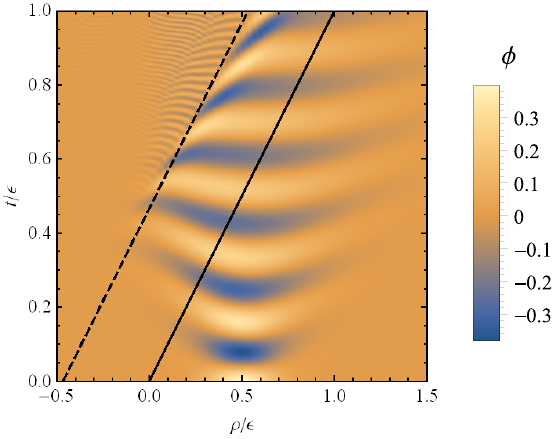}
\caption{Propagation of a scalar field near the inner horizon of a Hayward space-time. The dashed line corresponds to this horizon, and the solid line corresponds to the radial coordinate $r=(2 rs)^{1/3} l^{2/3}$. The parameters used in the simulation are $l/\epsilon = 10$, $r_s/\epsilon = 60$, $m\epsilon = 0.8$, and $\kappa \epsilon = 0.4$. }
\label{fig:Propagation_Hayward}
\end{figure}

\section{Conclusion}

\label{sec:conclusion}

In this paper, several ideas of quantum gravity have been explored, such as the possibility to encode space-time geometric distances in the correlation of a spin-like system, and to use a discretization scheme based on Monte-Carlo principles instead of a lattice-based discretization. These two ingredients offer interesting advantages compared to other common quantum gravity theories which are based on elementary areas (instead of distances) and/or simplicial lattices for the discretization scheme. The use of distances leads to a simpler connection with general relativity and our everyday world, for which elementary quantities are distances. The Monte Carlo method simplifies the discretization scheme of differential operators, and it avoids conceptual problems induced by the arbitrary choice of a lattice.

To construct a physically relevant theory, the starting point was to quantize a scalar field in curved space-time. Working with such a test system has led to a simple link between general ideas and a well-defined mathematical model. In this model, the gravitational field is described by a quantum system which is similar to an ensemble of spin-$\tfrac{1}{2}$ particles. At each spin corresponds a point of space. The geodesic length between two of these points is given by the quantum state through the relation $C_{nk}^2=\langle  \hat \sigma_+(x_n) \hat \sigma_-(x_k) +\hat \sigma_+(x_k) \rangle _{\Psi_G}^2 \propto p_{x_n}(x_k)$. This relation is a postulate of the theory, and its physical relevance has been investigated with numerical tests. Quantum states that describe various classical geometries have been found with a very high numerical precision. This is a first positive result supporting this very strong postulate. The numerical simulations have also shown the existence of two different regimes associated with two different length scales. Bellow a specific threshold, the classical description does not hold anymore and one obtains something that can be called "a quantum space-time". This is a nontrivial consequence of the model, which remotely resembles the quantization of areas and volumes in loop quantum gravity. In loop quantum gravity, the minimal scale is given by the lowest eigenvalues of area and volume operators~\cite{rovelli_quantum_2004,rovelli_covariant_2014,ashtekar2021short}. This can be interpreted as the existence of minimal building blocks that can be used to construct the space-time. However, in the theory of this paper, the interpretation is slightly different. The classical space-time is observed only if the distance between two points is sufficiently far from each other. If the density of points becomes too short, there is a threshold where a classical picture is not possible anymore and the resulting object cannot be interpreted as a space-like sub-manifold. This effect may have interesting consequences in black hole physics and cosmology. The relevance of the model has also been investigated from another point of view, which has consisted of recovering, by means of successive approximation, the classical theory of a scalar field. Each step of the procedure has been illustrated by numerical examples. These examples show clearly the different regimes of the theory, but they also demonstrate its computational efficiency.

In these few pages, many issues have been glossed over. Here is a (non-exhaustive) list of the points that shall be investigated:

\begin{itemize}
\item One of the most important ones concerns the explicit construction of the gravitational Hamiltonian $H_G$, which is compatible with the Einstein-field equation (at least in some limit cases, compatible with current experimental observations). Another issue concerns the choice of the probability density $p_{x_n}(x_k)$, which has been assumed to be a normal law. In fact, the form of the density may be fixed by the theory itself, by $\hat H_G$ and its eigenvectors. The use of a normal Law is interesting because at small distances we have $C_{nk}^2 \propto 1 - q_{a b} dx^a dx^b/(2 \epsilon ^2)$, with $dx$ a vector connecting $x_n$ to $x_k$. The relation with the metric tensor is therefore very simple and quite natural, but many different distributions can have locally this form and very different behaviors at large scales.
\item In this paper, it was sufficient to provide a quantum mechanical description of the 3-metric $q_{ab}$. However, this is not sufficient to describe the entire space-time manifold $\Mc M$, the extrinsic curvature $k_{ab}$ is also necessary, and it must be encoded somewhere in $\ket{\tilde \Psi_G}$.
\item  A similar study with a Dirac field instead of a scalar field could also be interesting to test the coupling operator $\hat C_{nk}$. This operator (see Eq.~\eqref{eq:def_operator_Cnk}) has been chosen because it provides interesting numerical results, but this choice is not unique. For a Dirac field, it could be more interesting to choose an operator of the form $\hat C_{nk} = \sum_{i,j \in\{x,y,z\}} w_{ij} \hat \sigma_i(x_n) \hat \sigma(x_k)$, with $w_{ij} \in \setR$, with $\hat \sigma_i(x_n)$ the Pauli matrices.
\item  Another point to clarify concerns the possibility of interpreting $\hat \sigma_i(x_n)$ as a quantized 3D reference frame. If such a link can be consistently made with the ideas of the paper, this may provide a very interesting connection with current theories of quantum reference frame.
\item Finally, a last point to investigate with deeper details, together with the construction of $H_G$, is the physical consequence of the use of $\ket{\tilde \Psi_G}= \ket{ \Psi_G}^{\otimes 2}$. This condition has a pure mathematical origin, it could be relaxed, such that $\ket{\tilde \Psi_G}= \ket{ \Psi_1} \otimes \ket{\Psi_2}$. The two states would correspond only in the classical limit.
\end{itemize}  

\section*{Acknowledgement}

The author acknowledges David Viennot for useful discussions on emergent gravity.

\section*{Data Availability}

Numerical codes and data can be provided upon reasonable request.

\appendix

\section{Normal law in a curved space-time}
\label{sec:Normal law in a curved space-time}

\begin{theorem}
Let Synge's world function be expressed in terms of Riemann normal coordinates~\cite{poisson_motion_2011,Brewin_2009}:
\[
 \sigma(x_n,x_k) = \frac{1}{2}\left(\delta_{ab} y^a y^b - \frac{\lambda^2}{3} R_{acbd} y^a y^b y^c y^d   +O(\lambda^3) \right)
\]
with the coordinate system defined such that $x_n$ is located at the origin, and $x_k$ located at $y^a$. $\lambda$ is a scale parameter that must be sufficiently small, but also sufficiently large in front of $\epsilon$, and $R_{abcd}$ is Riemann tensor evaluated at $y=0$. Then, we have,
\[
\int d^3y \sqrt{q}~\frac{\Mc N(R)}{(\sqrt{2 \pi} \epsilon)^3} e^{-\sigma/\epsilon^2} = 1 + O(\lambda^3)
\]
with $\Mc N(R) = 1+ R \epsilon^2 \lambda^2 /6 + O(\lambda^3) $, $R$ being the scalar curvature, and
\[
\int d^3y \sqrt{q}~y^a y^b \frac{\Mc N(R)}{(\sqrt{2 \pi} \epsilon)^3} e^{-\sigma/\epsilon^2} = \epsilon^2 + O(\epsilon^3) + O(\lambda^3).
\]
\end{theorem}

\begin{proof}
The proof is based in both cases on a series expansion in powers of $\lambda$, followed by the integration. To perform the calculation $\sqrt{q}$ must be also Taylor expanded. Following anew~\cite{Brewin_2009}, one get
\begin{equation}
\sqrt{q} = 1 - \frac{\lambda ^2}{6} R_{ab}y^a y^b + O(\lambda^3),
\end{equation}
with $R_{ab}$ the Ricci tensor. Plugging all together the expansions of $\Mc N(R), \sigma$ and $\sqrt{q}$, the integral becomes:
\begin{equation*}
\begin{split}
&\int d^3y \sqrt{q}~\frac{\Mc N(R)}{(\sqrt{2 \pi} \epsilon)^3} e^{-\sigma/\epsilon^2} \\
&= \int d^3y \frac{ e^{-\delta_{ab}y^a y^b/\epsilon^2}}{(\sqrt{2 \pi} \epsilon)^3} \left( 1 + \frac{\lambda^2}{6}(R \epsilon^2 - R_{ab} y^a y^b  \right. \\
& ~~ \left. + \frac{R_{abcd}}{\epsilon^2} y^a y^b y^c y^d )\right) + O(\lambda^3)
\end{split}
\label{eq:taylor_expansion_normal_coordinate_integral_pnk}
\end{equation*}
The integral has non zero contributions only if the powers in $y^a$ are all even (otherwise positive and negative positions cancel each other). This leads us to the following replacement rules:
\begin{align}
\int d^3y \frac{ e^{-\delta_{ab}y^a y^b/\epsilon^2}}{(\sqrt{2 \pi} \epsilon)^3} y^a y^b &= \epsilon^2 \delta^{a b} \\
\int d^3y \frac{ e^{-\delta_{ab}y^a y^b/\epsilon^2}}{(\sqrt{2 \pi} \epsilon)^3} (y^a)^2 (y^b)^2 &= \epsilon^4(3-2 \delta^{a b}).
\end{align}
Moreover, the Ricci tensor is symmetric and the Riemann tensor is skew-symmetric on some of its indices. As a consequence, one get $R_{ab} y^a y^b \rightarrow R_{ab}\delta^{ab} =R$ (the scalar curvature being evaluated at $y^a =0$), and $R_{abcd} y^a y^b y^c y^d \rightarrow 0$. As a result, the second order term in $\lambda $ in \eqref{eq:taylor_expansion_normal_coordinate_integral_pnk} is canceled, and the result is $1+O(\lambda^3)$.

For the second relation of the theorem, the starting point is the same, but now, the minimum number of powers in $y^a$ is at least one. Once again, the term with the Riemann tensor is canceled, and the only surviving terms are the ones with $y^ay^b$, $R \epsilon^2 \lambda ^2 y^a y^b/6 $ , and $R_{cb} \lambda ^2 y^a y^b y^c y^d/6$. The result of the integral, is thus $\epsilon^2 - \epsilon ^4 \lambda^2 ( R + 2 R^{ab})/6$. Then, this leads to a final expression of $\epsilon^2 + O(\epsilon^3) + O(\lambda^3)$.

\end{proof}

\section{Details on the numerical optimizations}
\label{sec:Details on the numerical optimizations}

In Sec.~\ref{sec:Quantum space-time}, the assumptions at the core of the quantum description of space-time is tested numerically by searching numerically quantum states that can describe a classical space-time.

As earlier commented in the main text, the goal is to determine the vector component of a quantum state that minimize a cost function $\Mc F$. Several cases study are considered, and results are plotted in Fig.~\ref{fig:convergence_approximation_scheme}.  In the first one (panel a), the locations of $N$ points are generated randomly in a cube of edge length $L/\epsilon$. In the second case (panel b), the number of points is fixed, and they are located at the vertices of a regular polyhedron (a cube or a 3-simplex), the edge length being varied. In the third case (panel c), the situation is similar to the first one, but the cube is replaced by a 3-sphere of radius $r$. In the last case (panel d), points a randomly sampled in a Hayward space-time~\cite{PhysRevLett.96.031103}. With these examples, flat and two kind of curved spaces are investigated. The first curved space of interest is a 3-sphere, for which it is easy to compute geodesic lengths. The second curved space has, contrary to the 3-sphere, a position dependent scalar curvature. Hence, the normalization coefficient of the probability distribution (see Eq.~\eqref{eq:assumption_probability_density}) depends on the position. 

The Hayward\footnote{A Hayward space-time describes a Black hole without singularity.} space-time is defined by the line element (written in Lemaitre-type coordinates)
\begin{equation}
ds^2 = -dt ^2 + \frac{r_s r^2}{r^3 + l^2 r_s} d\rho^2 + r^2 d\Omega^2,
\label{eq:modified_Hayward_metric}
\end{equation}
with $t$ the time coordinate, $\rho$ the radial coordinate, and $\Omega$ is a solid angle coordinate, $r_s$ correspond to the Schwarzschild radius, and $l$ is a length scale of the central body. $r(\rho,t)$ is a non-analytic function of $\rho$ and $t$, defined by $d\rho -dt = dr\left(\tfrac{r_s r^2}{r^3 + l^2 r_s}\right)^{-1/2}$. It can be easily evaluated numerically or approximated analytically using a Taylor expansion. Since the coordinates are related by a differential equation, an arbitrary choice of origin is required. Here, $\rho-t=0$ is chosen to correspond with $r= 2^{1/3} l^{2/3} r_s^{1/3}$. In this paper only a limited portion of the space-time is considered. In the case of Fig.~\ref{fig:convergence_approximation_scheme}, $t=0$ and $\rho/\epsilon\in[-1,2]$, $l/\epsilon = 0.5$, and $r_s/\epsilon = 3$.

To perform the calculations, several key points must be considered. First of all, the minimization of $\Mc F$ is performed with a two-step procedure. The first step is a global minimization with the algorithm JAYA~\cite{venkata_rao_jaya:_2016}. With this algorithm, convergence towards a minimum may not be achieved (mostly when the number of spins is large), but the resulting quantum state is used as an input for a second optimization using a gradient descent algorithm~\cite{bonnans2006numerical}.

The second point is that $C_{nk}^2 \propto p_{x_n}(x_k)$, and the proportionality coefficient is unknown. The value of this coefficient, namely $\alpha$, is fixed by hand in the numerical optimization. In principle it can be optimized simultaneously with the quantum state, but at the price of an important computational cost. A third important point concerns the computation of geodesic lengths. This can be done exactly in the simplest cases, but otherwise they are only estimated approximately.

\begin{itemize}
\item In the first two cases $\alpha$ is simply fixed to $\alpha =1$, and the geodesic length is given by the Euclidean distance. We thus have $\sigma = \Vert x_n -x_k \Vert^2/2$.

\item In the case of the 3-sphere, the space is curved, $\alpha = (1 + R \epsilon^2 \lambda^2/6)(\sqrt{2 \pi} \epsilon)^{-3}$, with  $\lambda = 3 r /5$ (fixed by hand, such that $\epsilon < \lambda < r$). Recall that $r$ is the radius of the sphere, and the scalar curvature is given by $R = 6/r^2$. The geodesic length is given by the arc length, so that $\sigma = r^2 \arccos (u_n.u_k/r^2)^2/2$, with $u_n$ the point position in $\setR^4$.

\item In the case of the Hayward space, $\alpha = (\sqrt{2 \pi} \epsilon)^{-3}$. The geodesic length is estimated using a second order expansion of $\sigma$~\cite{poisson_motion_2011,Brewin_2009}
\end{itemize}

\section{Covariant Monte-Carlo discretization}
\label{sec:covariant_theory}

In this appendix, the ideas introduced in Sec.~\ref{sec:discretization_scalar_field_Hamiltonian} are reused to define a covariant quantum theory of the scalar field. In such a framework, the action is discretized, and it can be inserted in a second step into a path-integral. The precise definition of the path integral is not discussed here, the following few lines focus on the discretization of the action.
The usual action for a scalar field is given by:
\begin{equation}
S = \int d^4x  \frac{\sqrt{-g}}{2} \left[ g^{\mu\nu} (\partial_\mu \phi)( \partial_\nu \phi) + m^2) \phi^2 \right]
\end{equation}
It can be discretized with a Monte-Carlo approach, as given in Sec.~\ref{sec:Monte_Carlo_approx}, but using the space-time metric $g$ instead of the metric $q$ of a 3D hypersurface. The nontrivial point would be to estimate the term $g^{\mu\nu} (\partial_\mu \phi)( \partial_\nu \phi)$.

One can notice that:
\begin{equation}
\begin{split}
\mathbb{E}[\phi^2]_x & = \int d^4y \sqrt{-g}~ p_x(y) \phi(x)^2\\
& \approx \phi(x)^2 + \frac{\epsilon^2}{2} \phi(x) \nabla _\mu \nabla^\mu  \phi(x) + \epsilon^2 (\partial^\mu \phi)( \partial_\mu \phi)
\end{split}
\label{eq:taylor exp_kinetic_term_lagrangian}
\end{equation}
The computation steps are basically the same as in Eq.~\eqref{eq:taylor_expand_mean_phi_square}, but using the Taylor expansion of $\phi(x)^2$ over space-time. Using the relation of Eq.~\eqref{eq:secon_order_approx_kinetic_energy_density}, generalized to 4D, it is possible to rewrite Eq.~\eqref{eq:taylor exp_kinetic_term_lagrangian} into
\begin{equation}
\mathbb{E}[\phi^2]_x \approx \epsilon^2 \left( \frac{1}{2}\phi(x) \mathbb{E}[\phi]_x +  (\partial^\mu \phi)( \partial_\mu \phi) \right),
\end{equation}
and thus,
\begin{equation}
(\partial^\mu \phi)( \partial_\mu \phi) \approx \frac{1}{\epsilon ^2} \mathbb{E}[\phi^2]_x - \frac{1}{2} \phi(x)\mathbb{E}[\phi]_x.
\end{equation}
As a consequence, with the Monte-Carlo approximation, the action is approximately
\begin{equation}
S \approx \frac{V}{2N} \sum_{n=1}^N  m^2\phi_n^2 + \frac{V}{N} \sum_{k=1}^N p_{x_n}(x_k)\left( \frac{1}{\epsilon^2} \phi_k^2 - \frac{1}{2}\phi_n \phi_k \right)
\end{equation}
One can also factorize by $V/N$, as in Eq.~\eqref{eq:hamiltonian_quantum_final}, and the final result is
\begin{equation}
\begin{split}
S \approx& \frac{V^2}{2N^2} \sum_{n=1}^N  \left( \sum_{k=1}^N p_{x_n}(x_k)\right)m^2\phi_n^2  \\
&+ \sum_{k=1}^N p_{x_n}(x_k)\left( \frac{1}{\epsilon^2} \phi_k^2 - \frac{1}{2}\phi_n \phi_k
\right).
 \end{split}
\end{equation}
Once again, this formula is evaluated in 4D, which means that $p_{x_n}(x_k)$ is a function of the geodesic length in $\Mc M$, not in $\Sigma$.

\section{Hamiltonian and equation of motion of discretized systems}
\label{sec:classical_hamiltonian}

In the main text of this paper, repeated references to the classical Hamiltonian and the classical equation of motion are made, without detailed explanations on the subject. A few details on the subject are reported in this appendix. 

The classical Hamiltonian and the classical equation of motions are already given in Sec.~\ref{sec:quantum_scalar_field}, in the case of the continuous field. However, the formula must be slightly adapted in the discretized versions.

In the continuous case, Hamilton’s equations are given by functional derivatives~\cite{gelfand_calculus_2000}. In the case of $\delta$-peacked variation of a field, a functional $I = \int d\mu(x) f(\phi(x))$ has for variation:
\begin{equation}
\frac{\delta I}{\delta \phi(y)} = \int d\mu(x) \frac{\partial f}{\partial \phi} \delta(x-y) = \frac{\partial f}{\partial \phi}(y).
\end{equation}
In the Monte-Carlo approach, the Lebesgue measure \footnote{for simplicity, a flat space is assumed} $\mu$ is replaced by a sum of Dirac distribution, i.e.
\begin{equation}
I \approx I' = \frac{V}{N} \sum_n \int d\mu(x) \delta(x-x_n) f(\phi(x)).
\end{equation}
The functional derivative is, therefore, \footnote{to get this result, a product of Dirac distributions must be carefully computed using the Colombeau algebra~\cite{colombeau_elementary_2011}.}
\begin{equation}
\frac{\delta I'}{\delta \phi(y)} = \frac{V}{N} \sum_n \frac{\partial f}{\partial \phi}(x_n) \delta(x_n -y).
\end{equation}
This is a distribution. To identify the variations of $I'$ to the variations of $I$, $\tfrac{\delta I'}{\delta \phi(y)}$ must be regularized. The regularization can be performed by different means. A simple one is to consider:
\begin{equation}
\frac{\delta I}{\delta \phi(y)}  = \frac{N}{V} \int_{\Mc B_y} d\mu(x) \frac{\delta I'}{\delta \phi(y)} 
\end{equation}
with $\Mc B_y$ a ball centered on $y$ with a radius sufficiently small, so that a single Dirac distribution is integrated. The two functional derivatives are thus proportional, with a factor $N/V$.

As a consequence, the proportionality coefficient is recovered in all the formulas adapted from the continuous case. Then, Hamilton’s equations for the discretized system are given by:
\begin{align}
\frac{d \phi_n}{dt} =  \frac{N}{V}\frac{\partial H}{\partial \Pi_n} \\
\frac{d \Pi_n}{dt} = - \frac{N}{V} \frac{\partial H}{\partial \phi_n} .
\end{align}
with $H$ the discretized Hamiltonian, and the Schrödinger equation is
\begin{equation}
\frac{d}{dt} \ket{\psi}= -\ii  \frac{N}{V} \hat H \ket{\psi}.
\end{equation}

Now, a few details on the classical version of Eq.~\eqref{eq:hamiltonian_quantum_final} are given.
The classical Hamiltonian is obtained by reversing the canonical quantization procedure, which means that $\hat a_n$ and $\hat a_n^\dagger$ are replaced by $\setC$ numbers $a_n$ and $a_n^\dagger$. The corresponding Hamiltonian $H$ is therefore a function on the complex phase space $\setC^N$. Equations of motion are then obtained by complexified Hamilton equations $ \ii d_t a_n =  \tfrac{N}{V}\partial_{a_n^\dagger} H$. We can also come back to the real representation of the system, using the change of variable introduced in Eq.~\eqref{eq:def_a} and \eqref{eq:def_a_dagger}. After the change of variable, the classical Hamiltonian reads:
\begin{equation}
\begin{split}
H =&\frac{V^2 }{2N^2} \sum_{n=1}^N  \left( \sum_{k=1}^n p_{x_n}(x_k)\right)\left(  \Pi_n^2 + \left[m^2 +\frac{1}{\epsilon^2} \right] \phi_n^2   \right) \\
& - \frac{1}{ \epsilon^2  } \left( \phi_n \sum_{k=1}^N p_{x_n}(x_k )\phi_k \right).
\end{split}
\end{equation}

%


\begin{thebibliography}{10}

\bibitem{klammer2008fermions}
Daniela Klammer and Harold Steinacker.
\newblock ``Fermions and emergent noncommutative gravity''.
\newblock
  \href{https://dx.doi.org/https://doi.org/10.1088/1126-6708/2008/08/074}{Journal
  of High Energy Physics {\bf 2008}, 074}~(2008).

\bibitem{steinacker2010emergent}
Harold Steinacker.
\newblock ``Emergent geometry and gravity from matrix models: an
  introduction''.
\newblock
  \href{https://dx.doi.org/https://doi.org/10.1088/0264-9381/27/13/133001}{Classical
  and Quantum Gravity {\bf 27}, 133001}~(2010).

\bibitem{Viennot_2021}
David Viennot.
\newblock ``Emergent gravity and d-brane adiabatic dynamics: emergent lorentz
  connection''.
\newblock \href{https://dx.doi.org/10.1088/1361-6382/ac337d}{Classical and
  Quantum Gravity {\bf 38}, 245004}~(2021).

\bibitem{viennot_fuzzy_2022}
David Viennot.
\newblock ``Fuzzy schwarzschild (2 + 1)-spacetime''.
\newblock \href{https://dx.doi.org/10.1063/5.0091364}{Journal of Mathematical
  Physics {\bf 63}, 082302}~(2022).

\bibitem{konopka_quantum_2008}
Tomasz Konopka, Fotini Markopoulou, and Simone Severini.
\newblock ``Quantum graphity: {A} model of emergent locality''.
\newblock \href{https://dx.doi.org/10.1103/PhysRevD.77.104029}{Phys. Rev. D
  {\bf 77}, 104029}~(2008).

\bibitem{caravelli_properties_2011}
Francesco Caravelli and Fotini Markopoulou.
\newblock ``Properties of quantum graphity at low temperature''.
\newblock \href{https://dx.doi.org/10.1103/PhysRevD.84.024002}{Phys. Rev. D
  {\bf 84}, 024002}~(2011).

\bibitem{quach_domain_2012}
James~Q. Quach, Chun-Hsu Su, Andrew~M. Martin, and Andrew~D. Greentree.
\newblock ``Domain structures in quantum graphity''.
\newblock \href{https://dx.doi.org/10.1103/PhysRevD.86.044001}{Phys. Rev. D
  {\bf 86}, 044001}~(2012).

\bibitem{cao_space_2017}
ChunJun Cao, Sean~M. Carroll, and Spyridon Michalakis.
\newblock ``Space from {Hilbert} space: {Recovering} geometry from bulk
  entanglement''.
\newblock \href{https://dx.doi.org/10.1103/PhysRevD.95.024031}{Phys. Rev. D
  {\bf 95}, 024031}~(2017).

\bibitem{bao_hilbert_2017}
Ning Bao, Sean~M. Carroll, and Ashmeet Singh.
\newblock ``The {Hilbert} space of quantum gravity is locally
  finite-dimensional''.
\newblock \href{https://dx.doi.org/10.1142/S0218271817430131}{Int. J. Mod.
  Phys. D {\bf 26}, 1743013}~(2017).

\bibitem{PhysRevD.97.086003}
ChunJun Cao and Sean~M. Carroll.
\newblock ``Bulk entanglement gravity without a boundary: Towards finding
  einstein's equation in hilbert space''.
\newblock \href{https://dx.doi.org/10.1103/PhysRevD.97.086003}{Phys. Rev. D
  {\bf 97}, 086003}~(2018).

\bibitem{padmanabhan2015emergent}
Thinakkal Padmanabhan.
\newblock ``Emergent gravity paradigm: recent progress''.
\newblock
  \href{https://dx.doi.org/https://doi.org/10.1142/S0217732315400076}{Modern
  Physics Letters A {\bf 30}, 1540007}~(2015).

\bibitem{andrea_mondino_optimal_2022}
Andrea Mondino and Stefan Suhr.
\newblock ``An optimal transport formulation of the einstein equations of
  general relativity''.
\newblock \href{https://dx.doi.org/https://doi.org/10.4171/jems/1188}{Journal
  of the European Mathematical Society {\bf 25}, 933--994}~(2022).

\bibitem{gorard_quantum_2020}
Jonathan Gorard.
\newblock ``Some quantum mechanical properties of the wolfram model''.
\newblock
  \href{https://dx.doi.org/http://dx.doi.org/10.25088/ComplexSystems.29.2.537}{Complex
  Syst.{\bf 29}}~(2020).

\bibitem{gorard_relativistic_2020}
Jonathan Gorard.
\newblock ``Some relativistic and gravitational properties of the wolfram
  model''.
\newblock \href{https://dx.doi.org/10.25088/ComplexSystems.29.2.599}{Complex
  Systems {\bf 29}, 599--654}~(2020).

\bibitem{universe5010035}
Lisa Glaser and Sebastian Steinhaus.
\newblock ``Quantum gravity on the computer: Impressions of a workshop''.
\newblock \href{https://dx.doi.org/10.3390/universe5010035}{Universe{\bf
  5}}~(2019).

\bibitem{ashtekar2021short}
Abhay Ashtekar and Eugenio Bianchi.
\newblock ``A short review of loop quantum gravity''.
\newblock
  \href{https://dx.doi.org/https://doi.org/10.1088/1361-6633/abed91}{Reports on
  Progress in Physics {\bf 84}, 042001}~(2021).

\bibitem{yosifov2021aspects}
Alexander~Y Yosifov.
\newblock ``Aspects of semiclassical black holes: Development and open
  problems''.
\newblock
  \href{https://dx.doi.org/https://doi.org/10.1155/2021/6628693}{Advances in
  High Energy Physics {\bf 2021}, 1--13}~(2021).

\bibitem{physics5010001}
Suddhasattwa Brahma, Robert Brandenberger, and Samuel Laliberte.
\newblock ``Bfss matrix model cosmology: Progress and challenges''.
\newblock \href{https://dx.doi.org/10.3390/physics5010001}{Physics {\bf 5},
  1--10}~(2023).

\bibitem{PhysRevLett.100.070502}
Michael~M. Wolf, Frank Verstraete, Matthew~B. Hastings, and J.~Ignacio Cirac.
\newblock ``Area laws in quantum systems: Mutual information and
  correlations''.
\newblock \href{https://dx.doi.org/10.1103/PhysRevLett.100.070502}{Phys. Rev.
  Lett. {\bf 100}, 070502}~(2008).

\bibitem{ambjorn_causal_2013}
Jan Ambjorn, Andrzej Görlich, Jerzy Jurkiewicz, and Renate Loll.
\newblock ``Causal dynamical triangulations and the search for a theory of
  quantum gravity''.
\newblock \href{https://dx.doi.org/10.1142/S021827181330019X}{Int. J. Mod.
  Phys. D {\bf 22}, 1330019}~(2013).

\bibitem{rovelli_quantum_2004}
Carlo Rovelli.
\newblock ``Quantum {Gravity}''.
\newblock \href{https://dx.doi.org/10.1017/CBO9780511755804}{Cambridge
  {Monographs} on {Mathematical} {Physics}}. Cambridge University Press.
  ~(2004).

\bibitem{rovelli_covariant_2014}
Carlo Rovelli and Francesca Vidotto.
\newblock ``Covariant loop quantum gravity: An elementary introduction to
  quantum gravity and spinfoam theory''.
\newblock
  \href{https://dx.doi.org/https://doi.org/10.1017/CBO9781107706910}{Cambridge
  University Press}. ~(2014).

\bibitem{christ_weights_1982}
N.~H. Christ, R.~Friedberg, and T.~D. Lee.
\newblock ``Weights of links and plaquettes in a random lattice''.
\newblock \href{https://dx.doi.org/10.1016/0550-3213(82)90124-9}{Nuclear
  Physics B {\bf 210}, 337--346}~(1982).

\bibitem{ren_matter_1988}
Hai-cang Ren.
\newblock ``Matter fields in lattice gravity''.
\newblock \href{https://dx.doi.org/10.1016/0550-3213(88)90281-7}{Nuclear
  Physics B {\bf 301}, 661--684}~(1988).

\bibitem{teixeira_random_2013}
F.~L. Teixeira.
\newblock ``Random {Lattice} {Gauge} {Theories} and {Differential}
  {Forms}''~(2013).
\newblock  url:~\url{http://arxiv.org/abs/1304.3485}.

\bibitem{yamamoto_lattice_2014}
Arata Yamamoto.
\newblock ``Lattice {QCD} in curved spacetimes''.
\newblock \href{https://dx.doi.org/10.1103/PhysRevD.90.054510}{Phys. Rev. D
  {\bf 90}, 054510}~(2014).

\bibitem{Brower:2016moq}
Richard~C. Brower, George Fleming, Andrew Gasbarro, Timothy Raben, Chung-I Tan,
  and Evan Weinberg.
\newblock ``{Quantum Finite Elements for Lattice Field Theory}''.
\newblock \href{https://dx.doi.org/10.22323/1.251.0296}{PoS {\bf LATTICE2015},
  296}~(2016).
\newblock  \href{http://arxiv.org/abs/1601.01367}{arXiv:1601.01367}.

\bibitem{caflisch_monte_1998}
Russel~E. Caflisch.
\newblock ``Monte carlo and quasi-monte carlo methods''.
\newblock \href{https://dx.doi.org/10.1017/S0962492900002804}{Acta Numerica
  {\bf 7}, 1--49}~(1998).

\bibitem{zappa2018monte}
Emilio Zappa, Miranda Holmes-Cerfon, and Jonathan Goodman.
\newblock ``Monte carlo on manifolds: sampling densities and integrating
  functions''.
\newblock
  \href{https://dx.doi.org/https://doi.org/10.48550/arXiv.1702.08446}{Communications
  on Pure and Applied Mathematics {\bf 71}, 2609--2647}~(2018).

\bibitem{de2018quasi}
Stefano De~Marchi and Giacomo Elefante.
\newblock ``Quasi-monte carlo integration on manifolds with mapped
  low-discrepancy points and greedy minimal riesz s-energy points''.
\newblock
  \href{https://dx.doi.org/https://doi.org/10.1016/j.apnum.2017.12.017}{Applied
  Numerical Mathematics {\bf 127}, 110--124}~(2018).

\bibitem{poisson_motion_2011}
Eric Poisson, Adam Pound, and Ian Vega.
\newblock ``The motion of point particles in curved spacetime''.
\newblock \href{https://dx.doi.org/10.12942/lrr-2011-7}{Living Reviews in
  Relativity {\bf 14}, 7}~(2011).

\bibitem{khavkine_algebraic_2015}
Igor Khavkine and Valter Moretti.
\newblock ``Algebraic {QFT} in curved spacetime and quasifree {Hadamard}
  states: an introduction''.
\newblock In Advances in algebraic quantum field theory.
\newblock Pages 191--251.
\newblock Springer~(2015).

\bibitem{gerard_introduction_2018}
Christian Gérard.
\newblock ``An introduction to quantum field theory on curved spacetimes''.
\newblock
  \href{https://dx.doi.org/https://doi.org/10.1017/9781108186612.004}{Page
  171–218}.
\newblock London Mathematical Society Lecture Note Series. Cambridge University
  Press. ~(2018).

\bibitem{misner_gravitation_1973}
Charles~W. Misner, Kip~S. Thorne, and John~Archibald Wheeler.
\newblock ``Gravitation''.
\newblock W. H. Freeman. ~(1973).
\newblock 1st edition.

\bibitem{itzykson_quantum_2012}
Claude Itzykson and Jean-Bernard Zuber.
\newblock ``Quantum field theory''.
\newblock Courier Corporation. ~(2012).

\bibitem{PhysRevB.40.546}
Matthew P.~A. Fisher, Peter~B. Weichman, G.~Grinstein, and Daniel~S. Fisher.
\newblock ``Boson localization and the superfluid-insulator transition''.
\newblock \href{https://dx.doi.org/10.1103/PhysRevB.40.546}{Phys. Rev. B {\bf
  40}, 546--570}~(1989).

\bibitem{dutta2015non}
Omjyoti Dutta, Mariusz Gajda, Philipp Hauke, Maciej Lewenstein, Dirk-S{\"o}ren
  L{\"u}hmann, Boris~A Malomed, Tomasz Sowi{\'n}ski, and Jakub Zakrzewski.
\newblock ``Non-standard hubbard models in optical lattices: a review''.
\newblock \href{https://dx.doi.org/10.1088/0034-4885/78/6/066001}{Reports on
  Progress in Physics {\bf 78}, 066001}~(2015).

\bibitem{mivehvar2021cavity}
Farokh Mivehvar, Francesco Piazza, Tobias Donner, and Helmut Ritsch.
\newblock ``Cavity qed with quantum gases: new paradigms in many-body
  physics''.
\newblock
  \href{https://dx.doi.org/https://doi.org/10.1080/00018732.2021.1969727}{Advances
  in Physics {\bf 70}, 1--153}~(2021).

\bibitem{arovas2022hubbard}
Daniel~P Arovas, Erez Berg, Steven~A Kivelson, and Srinivas Raghu.
\newblock ``The hubbard model''.
\newblock
  \href{https://dx.doi.org/https://doi.org/10.1146/annurev-conmatphys-031620-102024}{Annual
  review of condensed matter physics {\bf 13}, 239--274}~(2022).

\bibitem{aidelsburger2022cold}
Monika Aidelsburger, Luca Barbiero, Alejandro Bermudez, Titas Chanda, Alexandre
  Dauphin, Daniel Gonz{\'a}lez-Cuadra, Przemys{\l}aw~R Grzybowski, Simon Hands,
  Fred Jendrzejewski, Johannes J{\"u}nemann, et~al.
\newblock ``Cold atoms meet lattice gauge theory''.
\newblock
  \href{https://dx.doi.org/https://doi.org/10.1098/rsta.2021.0064}{Philosophical
  Transactions of the Royal Society A {\bf 380}, 20210064}~(2022).

\bibitem{BEMANI2018186}
F.~Bemani, R.~Roknizadeh, and M.H. Naderi.
\newblock ``Quantum simulation of discrete curved spacetime by the
  bose–hubbard model: From analog acoustic black hole to quantum phase
  transition''.
\newblock
  \href{https://dx.doi.org/https://doi.org/10.1016/j.aop.2017.11.010}{Annals of
  Physics {\bf 388}, 186--196}~(2018).

\bibitem{PhysRevLett.96.031103}
Sean~A. Hayward.
\newblock ``Formation and evaporation of nonsingular black holes''.
\newblock \href{https://dx.doi.org/10.1103/PhysRevLett.96.031103}{Phys. Rev.
  Lett. {\bf 96}, 031103}~(2006).

\bibitem{peng_introduction_1998}
Jin-Sheng Peng.
\newblock ``Introduction {To} {Modern} {Quantum} {Optics}''.
\newblock \href{https://dx.doi.org/https://doi.org/10.1142/3770}{Wspc}.
  Singapore~(1998).
\newblock 1st edition.

\bibitem{gardiner_quantum_2004}
Crispin Gardiner, Peter Zoller, and Peter Zoller.
\newblock ``Quantum noise: a handbook of {Markovian} and non-{Markovian}
  quantum stochastic methods with applications to quantum optics''.
\newblock Springer Science \& Business Media. ~(2004).

\bibitem{breuer_theory_2007}
Heinz-Peter Breuer and Francesco Petruccione.
\newblock ``The {Theory} of {Open} {Quantum} {Systems}''.
\newblock Oxford University Press. Oxford~(2007).

\bibitem{azouit_adiabatic_2016}
Remi Azouit, Alain Sarlette, and Pierre Rouchon.
\newblock ``Adiabatic elimination for open quantum systems with effective
  {Lindblad} master equations''~(2016).
\newblock  url:~\url{http://arxiv.org/abs/1603.04630}.

\bibitem{Bassi_2017}
Angelo Bassi, André Großardt, and Hendrik Ulbricht.
\newblock ``Gravitational decoherence''.
\newblock \href{https://dx.doi.org/10.1088/1361-6382/aa864f}{Classical and
  Quantum Gravity {\bf 34}, 193002}~(2017).

\bibitem{viennot2017adiabatic}
David Viennot and Olivia Moro.
\newblock ``Adiabatic transport of qubits around a black hole''.
\newblock
  \href{https://dx.doi.org/https://doi.org/10.1088/1361-6382/aa5b5c}{Classical
  and Quantum Gravity {\bf 34}, 055005}~(2017).

\bibitem{kramer_localization_1993}
B.~Kramer and A.~MacKinnon.
\newblock ``Localization: theory and experiment''.
\newblock \href{https://dx.doi.org/10.1088/0034-4885/56/12/001}{Rep. Prog.
  Phys. {\bf 56}, 1469--1564}~(1993).

\bibitem{sachdev1999quantum}
Subir Sachdev.
\newblock ``Quantum phase transitions''.
\newblock
  \href{https://dx.doi.org/https://doi.org/10.1017/CBO9780511973765}{Physics
  world {\bf 12}, 33}~(1999).

\bibitem{PhysRevA.76.032116}
E.~M. Graefe and H.~J. Korsch.
\newblock ``Semiclassical quantization of an $n$-particle bose-hubbard model''.
\newblock \href{https://dx.doi.org/10.1103/PhysRevA.76.032116}{Phys. Rev. A
  {\bf 76}, 032116}~(2007).

\bibitem{Veksler_2015}
Hagar Veksler and Shmuel Fishman.
\newblock ``Semiclassical analysis of bose–hubbard dynamics''.
\newblock \href{https://dx.doi.org/10.1088/1367-2630/17/5/053030}{New Journal
  of Physics {\bf 17}, 053030}~(2015).

\bibitem{Brewin_2009}
Leo Brewin.
\newblock ``Riemann normal coordinate expansions using cadabra''.
\newblock \href{https://dx.doi.org/10.1088/0264-9381/26/17/175017}{Classical
  and Quantum Gravity {\bf 26}, 175017}~(2009).

\bibitem{venkata_rao_jaya:_2016}
Ravipudi Venkata~Rao.
\newblock ``Jaya: {A} simple and new optimization algorithm for solving
  constrained and unconstrained optimization problems''.
\newblock \href{https://dx.doi.org/10.5267/j.ijiec.2015.8.004}{International
  Journal of Industrial Engineering Computations {\bf 7}, 19--34}~(2016).

\bibitem{bonnans2006numerical}
Joseph-Fr{\'e}d{\'e}ric Bonnans, Jean~Charles Gilbert, Claude Lemar{\'e}chal,
  and Claudia~A Sagastiz{\'a}bal.
\newblock ``Numerical optimization: theoretical and practical aspects''.
\newblock Springer Science \& Business Media. ~(2006).

\bibitem{gelfand_calculus_2000}
Izrail~Moiseevitch Gelfand, Richard~A Silverman, and {others}.
\newblock ``Calculus of variations''.
\newblock Courier Corporation. ~(2000).

\bibitem{colombeau_elementary_2011}
J.~F. Colombeau.
\newblock ``Elementary {Introduction} to {New} {Generalized} {Functions}''.
\newblock Elsevier. ~(2011).

\end{thebibliography}

\end{document}